\documentclass[journal]{IEEEtran}
\usepackage{bbm}
\usepackage{caption}
\usepackage{subfigure}
\usepackage{graphicx}
\usepackage{latexsym}
\usepackage{diagbox}
\usepackage{changepage}
\usepackage[fleqn]{amsmath}
\usepackage{amsmath}
\usepackage{amsfonts}
\usepackage{indentfirst}
\usepackage{CJK}
\usepackage{indentfirst}
\usepackage[varg]{txfonts}
\usepackage{stfloats}
\usepackage{multirow}%
\usepackage{booktabs}
\usepackage{color,soul}
\usepackage{epstopdf}
\usepackage{float}
\usepackage{bm}
\usepackage{makecell}
\usepackage{array}
\usepackage{bm}
\usepackage{mathtools}
\usepackage{geometry}
\usepackage[linesnumbered,ruled,commentsnumbered,longend]{algorithm2e}
\usepackage[linesnumbered,ruled,commentsnumbered,longend]{algorithm2e}
\makeatletter

\newcommand{\Rmnum}[1]{\expandafter\@slowromancap\romannumeral #1@}

\makeatother
\makeatletter
\newtheorem{theorem}{Theorem}

\newenvironment{proof}[1][Proof]{\begin{trivlist}
		\item[\hskip \labelsep {\itshape #1}]}{\end{trivlist}}

\newcommand{\qed}{\nobreak \ifvmode \relax \else
	\ifdim\lastskip<1.5em \hskip-\lastskip
	\hskip1.5em plus0em minus0.5em \fi \nobreak
	\vrule height0.75em width0.5em depth0.25em\fi}

\SetAlgoLongEnd

\begin{document}

\title{Energy-Efficient Hybrid Precoding Design for Integrated Multicast-Unicast Millimeter Wave Communications with SWIPT}
\author{Wanming Hao,~\IEEEmembership{Member,~IEEE,}  Gangcan Sun, Fuhui Zhou,~\IEEEmembership{Member,~IEEE,} De Mi, Jia Shi,~\IEEEmembership{Member,~IEEE,} Pei Xiao,~\IEEEmembership{Senior Member,~IEEE} and Victor  C. M. Leung,~\IEEEmembership{Fellow,~IEEE}
    \thanks{W. Hao and G. Sun are with the School of Information Engineering, Zhengzhou University, Zhengzhou 450001, China, and W. Hao is also with the 5G Innovation Center, Institute of Communication Systems, University of Surrey, Guildford GU2 7XH, U.K. (E-mail: \{iewmhao, iegcsun\}@zzu.edu.cn).}
	\thanks{D. Mi and P. Xiao are with  the 5G Innovation Center, Institute of Communication Systems, University of Surrey, Guildford GU2 7XH, U.K. (Email: \{d.mi, p.xiao\}@surrey.ac.uk).}
	\thanks{F. Zhou is the College of Electronic and Information Engineering, Nanjing University of 
		Aeronautics and Astronautics, Nanjing, 210000, P. R. China. (Email: zhoufuhui@ieee.org). 
	}
    \thanks{J. Shi is with the State Key Laboratory of Integrated Services Networks, Xi'dian University, Xian, 710071, China (email:jiashi@xidian.edu.cn).}
	\thanks{Victor C. M. Leung is with College of Computer Science and Software Engineering, Shenzhen University, Shenzhen, China 518060, and the Department of Electrical and Computer Engineering, The University of British Columbia, Vancouver, BC V6T1Z4, Canada (e-mail: vleung@ece.ubc.ca).}
}
\maketitle

\begin{abstract}In this paper, we investigate the energy-efficient hybrid  precoding design for integrated multicast-unicast millimeter wave (mmWave) system, where the simultaneous wireless information and power transform is considered at receivers. We adopt two sparse radio frequency chain antenna structures at the base station (BS), i.e., fully-connected and subarray structures, and design the codebook-based analog precoding according to the different structures. Then, we formulate a joint digital multicast, unicast precoding and power splitting ratio optimization problem to maximize  the energy efficiency of the system, while the maximum transmit power at the BS and minimum harvested energy at receivers are considered. Due to its difficulty to directly solve the formulated problem, we equivalently transform the fractional objective function  into a subtractive form one and propose a two-loop iterative algorithm to solve it. For the outer loop, the classic Bi-section iterative algorithm is applied. For the inner loop,  we transform the formulated problem into a convex one by successive convex approximation techniques and propose  an iterative algorithm to solve it. Meanwhile, to reduce the complexity of the inner loop, we develop a zero forcing (ZF) technique-based low complexity iterative algorithm. Specifically, the ZF technique is applied to cancel the inter-unicast interference and the first order Taylor approximation is used for the convexification of the non-convex constraints in the original problem. Finally, simulation results are provided to compare the performance of the proposed algorithms under different schemes.
\end{abstract}


\section{Introduction}
Millimeter wave (mmWave) (30-300 GHz), owning the wider bandwidth, has been considered as a promising technique to meet the requirement with an exponential data traffic growth in future wireless communications~\cite{Ref0}. Furthermore, due to the short wavelengths in mmWave bands, more antennas can be packed with a small physical size, forming a massive multi-input multi-output (mMIMO) mmWave system~\cite{Ref2}.  However, the use of a large number of antennas will cause a huge energy consumption and hardware cost when fully digital signal processing is applied, in the sense that each antenna needs a dedicate radio frequency (RF) chain, and the power consumption of the RF chain is as high as 250 mW at mmWave frequencies~\cite{Ref3}. To tackle this problem, a hybrid analog/digital precoding scheme can be employed, where the required number of RF chains will be much less than that of antennas~\cite{Ref4}.  Based on the connectivity of RF chains, two types of structures are generally considered when the number of RF chains is small, one is fully-connected structure and the other is subarray structure. For the former, each RF chain is connected to all the antennas with a large number of phase shifters, which can obtain a high spectrum efficiency (SE). On the contrary, for the latter, each RF chain is required to connect a subset of antennas with a small number of phase shifters, thus a high energy efficiency (EE) is obtained~\cite{Ref5}-\cite{Ref7}. {For example, \cite{8Ge} investigates the energy and cost efficiency optimization solutions for 5G wireless communication systems with massive antenna and RF chains, and formulates a EE optimization problem. Based on this, a suboptimal iterative hybrid precoding algorithm is proposed. In~\cite{9Ge}, the authors formulate a joint optimization problem of computation and communication power based on a partially-connected RF chain structure and an upper bound of EE is derived. Meanwhile, a suboptimal solution consisting of the baseband and RF precoding matrices is proposed.}

On the other hand, the simultaneous wireless information and power transfer (SWIPT) has also been identified a promising technique for future wireless communications~\cite{Ref8}-\cite{Ref10}. In general, there are two practical schemes for the SWIPT, namely power splitting and time switching~\cite{Ref8}. With power splitting, the receivers split the received RF signals for information detection and  energy harvesting at the same time, while with time switching, the receivers switch between information detection and energy harvesting at different times~\cite{Ref11}.  In fact, SWIPT is a very effective solution for a multi-user system, where the interference power can be transformed into the energy at receivers. For example, \cite{Ref12},~\cite{Ref14} study the joint information and energy beamforming optimization problem in a MIMO interference channels, and some interference alignment schemes are proposed in the interference networks with SWIPT~\cite{Ref15},~\cite{Ref16}. However, the interference channel also causes  the difficulty for information decoding~\cite{Ref17}. Therefore, how to trade off the information and harvested energy is a challenge in the SWIPT system.

Having both mmWave and SWIPT as technological enablers for the energy-efficient wireless communications, the future cellular network can potentially support a wide range of services with a variety of very diverse requirements. There is an increasing demand for the multicast content delivery service over cellular networks, where a group of subscribed users intend to receive the same content~\cite{Ref18},~\cite{Ref19}. It is often that these users would request the customized content, when consuming the multicast content at the same time. Taking the object-based broadcasting (OBB) scenario as an example, each of the subscribed users intends to simultaneously receive both common message over multicast and private message over unicast. To this end, a joint multicast and unicast transmission can be an effective and efficient solution approach, comparing with the conventional frequency/time division multiplexing~\cite{Ref20}. However, the cross interference between multicast and unicast should be properly managed~\cite{Ref21}. {Therefore, the major challenges in considering joint unicast and multicast include: $\textit{i)}$ How to simultaneously transmit unicast and multicast signals; $\textit{ii)}$ How to jointly design the unicast and multicast beamforming.}

\begin{figure}[t]
	\begin{center}
		\includegraphics[width=9cm,height=3cm]{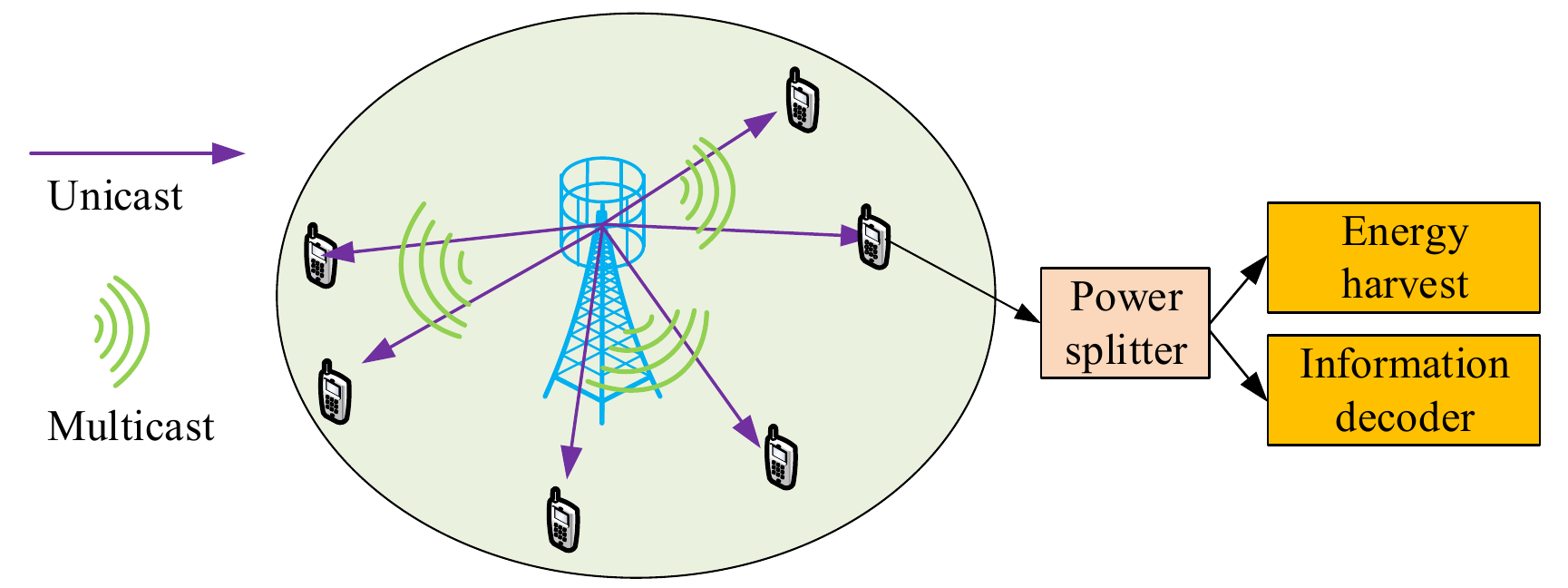}
		\caption{Joint multicast-unicast millimeter wave massive MIMO systems with~SWIPT.}
		\label{figsystem}
	\end{center}
\end{figure}

\subsection{Related Works}
Currently, several works have been conducted to investigate the part of the above problems. For the SWIPT system, \cite{Ref21a} studies the joint optimal power allocation and power splitting ratio to maximize the minimum signal to interference plus noise ratio (SINR) of all users. The authors in \cite{Ref21b} investigate the EE maximization problem by jointly optimizing the beamforming and power splitting ratio. Then, a low-complexity zero forcing (ZF) beamforming algorithm is proposed. Later, the authors extend the SWIPT system to heterogeneous networks~\cite{Ref21c}, where small cell BSs can harvest the energy from the macro BS. Based on this, a joint optimization energy harvesting rate and achievable throughput of small cell users is formulated, and then a sub-optimal iterative algorithm is proposed. The authors in~\cite{Ref21d} apply the SWIPT to a multicast system, where multiple users share the same message. The authors propose an efficient subcarrier allocation and power allocation scheme to maximize the minimum SINR at each subcarrier. Similarly, they later formulate a non-convex optimization problem as maximizing the minimum SINR among users in~\cite{Ref24} , and two successive convex approximation (SCA)-based iterative algorithms are proposed to solve the formulated  problem. Meanwhile, they extend the system to the sparse RF chain structure at the BS in~\cite{Ref25}. For the above system, they develop an efficient antenna selection and hybrid beamforming design algorithm to minimize the transmit power.

In addition, \cite{Ref26} investigates the spectrum and energy-efficient beamforming design problem in the mMIMO-NOMA mmWave with lens antenna array, and the ZF precoding scheme is used to reduce the inter-beam interference. Then, a dynamic power allocation algorithm is proposed to maximize the sum rate of the system. Later, the authors extend the system to SWIPT in~\cite{Ref26a}, and they propose an effective hybrid precoding and user grouping algorithm. Next, the weighted minimum mean square error (WMMSE)-based power allocation algorithm is developed to solve the formulated sum rate maximization problem.   The authors in~\cite{Ref22a}  investigate the energy-efficient hybrid beamforming strategy, where the fully digital power minimization problem is first analyzed. Then, the authors propose an iterative hybrid beamforming scheme to obtain the optimal solution. In~\cite{Ref23}, the authors investigate the hybrid beamforming design problem in mmWave joint unicast and multicast system. Based on this, a low-complexity optimization algorithm is proposed to solve the formulated sum rate maximization problem.

However,  the existing works have not jointly investigated the multicast-unicast mmWave communication with SWIPT. For example, \cite{Ref21a}-\cite{Ref21c} only consider the SWIPT, and~\cite{Ref21d}-\cite{Ref25} only study the multicast transmission. In addition, \cite{Ref26}-\cite{Ref23} investigate the hybrid precoding in mmWave, and the joint multicast and unicast transmission is only considered in~\cite{Ref23}.
\subsection{Main Contributions}
Different from the previous works, in this paper, we consider the joint multicast-unicast mmWave communication, where the SWIPT is applied at each receiver.  The main contributions are summarized as follows
\begin{itemize}
	{\item To reduce the hardware cost and energy consumption, we consider two sparse RF chain structures at the BS, i.e, fully-connected and subarray structures, and the corresponding analog precoding  are designed according to the predefined codebook. On this basis, we formulate an EE maximization problem by jointly optimizing unicast, multicast precoding and power splitting ratio. Meanwhile, we consider the maximum transmit power constraint of the BS and the minimum harvested energy requirement of each receiver. The formulated problem is non-convex, which is intractable in its original form. 
	\item We equivalently transform the fractional objective function of the optimization problem into a subtractive form one and then,  a two-loop iterative algorithm is developed. Specifically, the Bi-section iterative algorithm is applied at the outer loop. For the inner loop, we still need to solve a non-convex optimization problem. To this end, by introducing auxiliary variables and employing the SCA technique, we transform the original  problem into a convex one and propose an iterative algorithm to solve it.
	\item To reduce the complexity of the inner loop, we further develop  a low-complexity iterative algorithm. Specifically, we apply the ZF technique to cancel the inter-unicast interference, simplifying the unicast beamforming design as power allocation problem. Then, the first Taylor approximation is adopted to transform the non-convex constraints into convex ones, and the convex optimization technique is used to solve the inner optimization problem.}
\end{itemize}

The rest of this paper is organized as follows. The system description and analog precoding design schemes are presented in Section~II. The EE maximization problem is formulated and solved in Section~III. A low complexity ZF-based iterative algorithm is developed in Section~IV. Simulation results are presented in Section~V. Finally, we conclude this paper in Section~VI.

\textit{Notations}: We use the following notations throughout this paper: $(\cdot)^\ast$, 
$(\cdot)^T$ and $(\cdot)^H$ denote the conjugate, transpose and Hermitian transpose, respectively, $\|\cdot\|$ is the Frobenius norm, ${\mathbb{C}}^{x\times y}$ means the space of $x\times y$ complex matrix, {Re($\cdot$)} denotes real number operation.

\section{System Description and Analog Precoding Design}
In this section, we will describe the investigated mmWave system, including system model and  mmWave channel model. Then, two analog precoding design schemes are developed. 
\subsection{System Description}
We consider a downlink mmWave communication system as shown in~Fig.~\ref{figsystem}, where the BS is equipped with $N_{\rm{TX}}$ antennas. To reduce the hardware cost and energy consumption, we assume that the BS is equipped with $N_{\rm{RF}}$ RF chains ($N_{\rm{RF}}\leq N_{\rm{TX}}$). In addition,  $K $ ($K\leq N_{\rm{RF}}$) single-antenna users are served simultaneously with multicast and unicast, where $\mathcal{K}=\{1,\dots,K\}$ denotes the set of all users.  In this paper, we focus on a scenario that all user receive a common information stream by multicast, meanwhile each user obtains a private information stream by unicast. In general, there are two types of sparse RF chain structures at the BS. One is the fully-connected structure as shown in~Fig.~\ref{a}, where each RF chain is connected to all antennas through $N_{\rm{TX}}$ phase shifters. Another is the subarray structure as shown in~Fig.~\ref{b}, where each RF chain is connected to a disjoint subset of antennas through several phase shifters.

The received signal by the $k$th user can be expressed as 
\begin{eqnarray}
\begin{aligned}
y_{k}=&\underbrace{{\bf{h}}_k{\bf{F}}{\bf{v}}_kx_k}_{\rm{Desired\;private\;signal}}+\underbrace{{\bf{h}}_k{\bf{F}}{\bf{v}}_0x_0}_{\rm{Desired\;common\;signal}}\\
+&\underbrace{{\bf{h}}_k{\bf{F}}\sum_{i\neq k}^{K}{\bf{v}}_ix_i}_{\rm{Multi-user\;interference}}+\underbrace{n_k}_{\rm{Noise}},
\end{aligned}
\end{eqnarray}
where ${\bf{h}}_k\in \mathbb{C}^{1\times N_{\rm{TX}}}$, ${\bf{v}}_k\in \mathbb{C}^{N_{\rm{RF}}\times 1}$, and $x_k$, respectively, denote the downlink channel vector, digital precoding vector and private signal for the $k$th user. ${\bf{v}}_0\in \mathbb{C}^{N_{\rm{RF}}\times 1}$ and $x_0$ are the digital precoding vector and common signal for the $k$th user, respectively. $n_k$ is an independent and identically distributed (i.i.d.) additive white Gaussian noise (AWGN) defined as  $\mathcal{CN}(0,\delta_0^2)$. ${\bf{F}}\in \mathbb{C}^{N_{\rm{TX}}\times N_{\rm{RF}}}$ means the analog precoding matrix implemented by the equal power splitter and phase shifters~\cite{hao}. For the fully-connected structure, ${\bf{F}}$ can be written as 
\begin{eqnarray}
	{\bf{F}}=[{\bf{f}}_1,{\bf{f}}_2,\dots,{\bf{f}}_{N_{\rm{RF}}}],
\end{eqnarray} 
where ${\bf{f}}_k\in \mathbb{C}^{N_{\rm{TX}}\times 1}$ is the analog precoding vector associated with the $k$th RF chain, and $|({\bf{f}}_k)_i|=1/\sqrt{N_{\rm{TX}}} (i\in \{1,\dots,N_{\rm{TX}}\})$. Similarly, for the subarray structure, ${\bf{F}}$ can be expressed as
\begin{eqnarray}\label{eq1b}
{\bf{F}}=\left[
\begin{array}{cccc}
{\bf{f}}_1 & {\bf{0}}&\cdots & {\bf{0}} \\
{\bf{0}}&{\bf{f}}_2 & \cdots & {\bf{0}}\\
\vdots & \vdots & \ddots&\vdots \\
{\bf{0}} & {\bf{0}} &\cdots & {\bf{f}}_{N_{\rm{RF}}} \\
\end{array}
\right],
\end{eqnarray}
where ${\bf{f}}_k\in\mathbb{C}^{N_{\rm{SUB}}\times 1}$ denotes the analog precoding vector associated with the $k$-th RF chain with $|({\bf{f}}_k)_i|\!=\!1/\sqrt{N_{\rm{SUB}}}\;(i\!=\!{1\!,\!\ldots\!,\!N_{\rm{SUB}}})$. Here, $N_{\rm{SUB}}$ denotes the number of antennas connected to each RF chain, and we assume that $N_{\rm{SUB}}$ is the same for all RF chains with $N_{\rm{SUB}}=N_{\rm{TX}}/N_{\rm{RF}}$ \footnote[1]{Here, $N_{\rm{SUB}}$ should be an integer. In fact, when $N_{\rm{SUB}}$ is not an integer, i.e. the number of antennas in different subarrays may be different, this  is also suitable in our scheme.}. 
\begin{figure}[t]
	\center
	\subfigure[]{
		\label{a} 
		\includegraphics[scale=0.8]{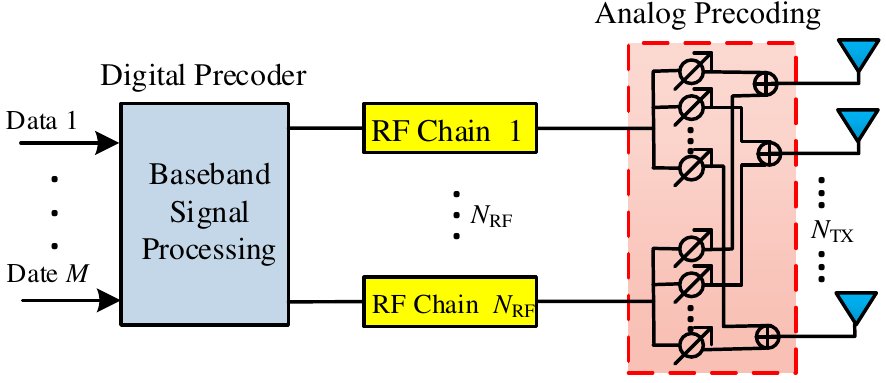}}
	\subfigure[]{
		\label{b} 
		\includegraphics[scale=0.8]{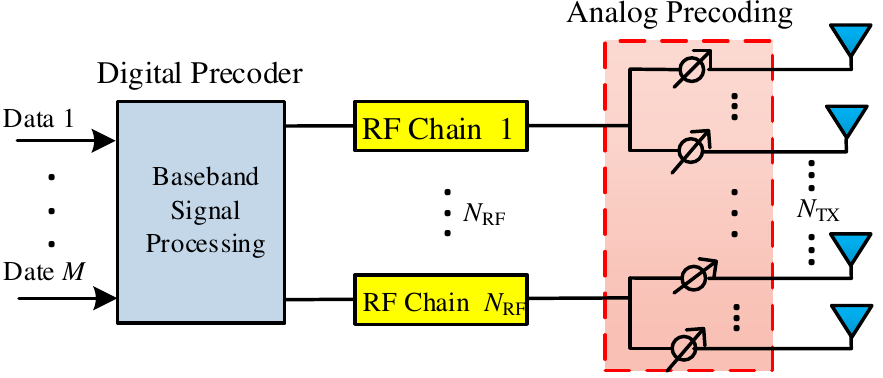}}
	\caption{Two sparse RF chain structures: (a) Fully-connected structure; (b) Subarray structure.}
	\label{SystemFigure} 
	\center
\end{figure}

In addition, each user splits the received signal into information decoder (ID) and energy harvester (EH). We assume that the $\beta_{k}$ portion of received signal power is divided into the ID, while the remaining $1-\beta_{k}$ portion of the received signal power is transformed into the EH. Accordingly, the received signal used for EH by the $k$th user can be written as
\begin{eqnarray}
	y_{k}^{\rm{EH}}=\sqrt{1-\beta_{k}}y_k,
\end{eqnarray}
and the harvested energy is
\begin{eqnarray}
	E_k=\varepsilon(1-\beta_{k})\left(\sum_{i=0}^K|{\bf{h}}_k{\bf{F}}{\bf{v}}_i|^2+\delta_0^2\right),
\end{eqnarray}
where $\varepsilon\in(0,1]$ denotes the energy conversion efficiency. The received signal used for ID can be expressed as
\begin{eqnarray}
	y_{k}^{\rm{ID}}=\sqrt{\beta_{k}}y_k+n'_k,
\end{eqnarray}
where $n'_k\sim \mathcal{CN}(0,\delta_1^2)$ is the addition noise caused by the ID. 

{Recently, layered-division multiplexing (LDM), a form of non-orthogonal multiplexing technology, has been introduced in cellular networks for joint multicast and unicast transmission, which is a key technology for next generation terrestrial digital television standard ATSC 3.0~\cite{24ldm}.  A layered transmission structure is applied to LDM for transmitting multiple signals with different power levels and robustness for different services and reception environment. Each user first decodes the upper layer most robust signal, and then cancels it from the received signal before decoding the next layer signal by successive interference cancellation (SIC) technique. Therefore, in this work, we can adopt a two-layer LDM structure, where the first layer is intended for multicast services and the second layer is for unicast services.   In general, the decoding order of the multicast and unicast messages at each receiver can be optimized according to the instantaneous channel condition. Since the multicast message is intended for multiple users and should have a higher priority~\cite{29Tao}, each user first decodes the multicast message and subtracts it from the received message, and then decodes the unicast message.} To this end, the achievable SINR of the common signal at the $k$th user can be expressed~as
\begin{eqnarray}\label{sinr1}
	\gamma_k^0=\frac{\beta_{k}|{\bf{h}}_k{\bf{F}}{\bf{v}}_0|^2}{\beta_{k}\left(\sum_{i=1}^K|{\bf{h}}_k{\bf{F}}{\bf{v}}_i|^2+\delta_0^2\right)+\delta_1^2},
\end{eqnarray} 
and the achievable SINR of the private signal at the $k$th user can be expressed as
\begin{eqnarray}\label{sinr2}
	\gamma_k=\frac{\beta_{k}|{\bf{h}}_k{\bf{F}}{\bf{v}}_k|^2}{\beta_{k}\left(\sum_{i\neq k}^K|{\bf{h}}_k{\bf{F}}{\bf{v}}_i|^2+\delta_0^2\right)+\delta_1^2}.
\end{eqnarray}

For mmWave channel, we adopt a widely used geometric channel mode as follows~\cite{28Heath}, 
\begin{eqnarray}
	{\bf{h}}_k=\sqrt{\frac{N_{\rm{TX}}}{L}}\sum_{l=1}^{L}\alpha_k^l{\bf{a}}(\theta_k^l),
\end{eqnarray}
where $L$ is the number of paths, $\alpha_k^l$ represents the complex gain of the $l$th path. ${\bf{a}}(\theta_k^l)$ is the antenna array response vector at user $k$. When the uniform linear array is used, ${\bf{a}}(\theta_k^l)$ can be expressed as
\begin{eqnarray}\label{array}
	{\bf{a}}(\theta_k^l)=\frac{1}{\sqrt{N_{\rm{TX}}}}\left[1,e^{j(2\pi /\lambda)d\sin(\theta_k^l)},\dots,e^{j(N_{\rm{TX}}-1)(2\pi /\lambda)d\sin(\theta_k^l)}\right],
\end{eqnarray}
where $\theta_k^l\in [0,2\pi]$ is the azimuth angles of departure of the BS at the $l$th path, $\lambda$ is the wavelength, and $d$ is the distance between two adjacent antenna array elements.

\begin{algorithm}[t]
	{\caption{The Analog Precoding Selection Algorithm for the Fully-Connected Structure.}
		\label{algorithm1}
		{\bf{Initialize}}  $\mathcal{A}$, ${\bf{F}}=[]$, $n=1$.\\
		\While {$n\leq N_{\rm{RF}}$}
		{\For{$k=1:K$}{Compute ${\bf{f}}_n^{\star}=\underset{{\bf{f}}_n\in\mathcal{A}}{\arg \max}\;\;|{\bf{h}}_k{\bf{f}}_n|^2,$\\
				${\bf{F}}=[{\bf{F}}\;{\bf{f}}_n^{\star}]$, $\mathcal{A}=\mathcal{A}-\{{\bf{f}}_n^{\star}\}$,
				$n=n+1$.\\
				\If{$n>N_{\rm{RF}}$}{Break;}}
	}}
\end{algorithm} 
\subsection{Analog Precoding Design}

{For the sparse RF chain structure, the design of the analog precoding depends on the phase shifters. Although several analog percoding methods~\cite{29Yu}-\cite{31Liang} that may outperform the codebook-based one, to decrease the complexity of joint design,  we first apply the codebook-based method for designing analog precoding~\cite{31Heath},~\cite{32Wang}.} Based on this, we can obtain the analog precoding by searching the codebook defined as $\mathcal{A}=\{{\bf{a}}(\theta_k^l),\forall k,l\}$. For the fully-connected structure, the analog precoding of the $k$th user can be selected as 
\begin{eqnarray}
{\bf{f}}_k^{\star}=\underset{{\bf{f}}_k\in\mathcal{A}}{\arg \max}\;\;|{\bf{h}}_k{\bf{f}}_k|^2,
\end{eqnarray}
and we summarize the analog precoding selection scheme as Algorithm~\ref{algorithm1}.

For the subarray structure, all RF chains are connected to the disjoint subset of antennas, and each of them has one beamforming direction. To this end, we divide the channel ${\bf{h}}_k$~as 
\begin{eqnarray}
{\bf{h}}_k=[{\bf{h}}_{k,1},{\bf{h}}_{k,2},\dots,{\bf{h}}_{k,N_{\rm{RF}}}],
\end{eqnarray} 
where ${\bf{h}}_{k,i}$ denotes the channel gain from the $i$th subarray antennas to the $k$th user. Different from the fully-connected structure, we need to search the codebook based on the subarray. For example, the analog precoding for the subarray $i$ at the $k$th user can be selected as 
\begin{eqnarray}
{\bf{f}}_k^{\star}=\underset{{\bf{f}}_k\in\mathcal{A}}{\arg \max}\;\;|{\bf{h}}_{k,i}{\bf{f}}_k|^2,
\end{eqnarray}
where $\mathcal{A}$ has become subarray-based codebook. Based on this, we summarize the analog precoding selection scheme as Algorithm~\ref{algorithm2}.
\begin{algorithm}[t]
	{\caption{The Analog Precoding Selection Algorithm for the Subarray Structure.}
		\label{algorithm2}
		{\bf{Initialize}}  $\mathcal{A}$, ${\bf{F}}={\bf{0}}_{N_{\rm{TX}}\times N_{\rm{RF}}}$, $n=1$.\\
		\While {$n\leq N_{\rm{RF}}$}
		{\For{$k=1:K$}{Compute ${\bf{f}}_n^{\star}=\underset{{\bf{f}}_n\in\mathcal{A}}{\arg \max}\;\;|{\bf{h}}_{k,n}{\bf{f}}_n|^2,$\\
				${\bf{F}}((n-1)N_{\rm{SUB}}+1:nN_{\rm{SUB}},n)={\bf{f}}_n^{\star}$, $\mathcal{A}=\mathcal{A}-\{{\bf{f}}_n^{\star}\}$,
				$n=n+1$.\\
				\If{$n>N_{\rm{RF}}$}{Break;}}
	}}
\end{algorithm} 
\section{EE Optimization Problem Formulation and Solution}
In this section, we formulate an EE maximization problem by jointly optimizing the unicast, multicast precoding and power splitting ratio, and then a two-loop iterative algorithm is proposed.
\subsection{Problem Formulation} 
In general, the power consumption includes two parts, namely transmit power and circuit power consumption. The circuit power consumption are mainly caused by baseband signal processing, RF chains and phase shifters~\cite{Ref2},~\cite{Ref6}. For the fully-connected structure, the circuit power consumptions can be written as
\begin{eqnarray}
	P_{\rm{C}}=P_{\rm{BB}}+N_{\rm{RF}}P_{\rm{RF}}+N_{\rm{RF}}N_{\rm{TX}}P_{\rm{PS}},
\end{eqnarray} 
where $P_{\rm{BB}}$, $P_{\rm{RF}}$, and $P_{\rm{PS}}$, respectively, denote the power consumption of the baseband, the RF chain and the phase shifter. Similarly, the circuit power consumption of the subarray structure can be expressed as
\begin{eqnarray}
P_{\rm{C}}=P_{\rm{BB}}+N_{\rm{RF}}P_{\rm{RF}}+N_{\rm{TX}}P_{\rm{PS}}.
\end{eqnarray}

Finally, we give the total power consumption as follows
\begin{eqnarray}
	P_{\rm{total}}=\sum_{k=0}^{K}\xi||{\bf{F}}{\bf{v}}_k||^2+P_{\rm{C}}.
\end{eqnarray}
where $\xi \geq 1$ is the inefficiency of the power amplifier~\cite{EE}.

Next, we define the EE of the system as
\begin{eqnarray}
	\eta_{\rm{EE}}\!=\!\frac{\underset{\forall k}{\min}\{\log_2(1\!+\!\gamma_k^0)\}\!+\!\sum_{k=1}^{K}\log_2(1\!+\!\gamma_k)}{P_{\rm{total}}} [\rm{bps/Hz}/\rm{W}].
\end{eqnarray}

In this paper, we aim  to maximize the EE of the system by jointly optimizing the power splitting ratio $\beta_{k}\;(k\in\mathcal{K})$ and digital precoding ${\bf{v}}_k\;(k\in\{0,\mathcal{K}\})$, which is written as
\begin{subequations}\label{OptA}
	\begin{align}
	\;\;\;\;\;\;\;\;\;\;\;\;\;\;\;\;\;\;\;\;\;\;\;&\underset{\left\{\{\beta_{k}\},\;\{{\bf{v}}_k\}\right\}}{\rm{max}}\;\;\;\;\;\;\eta_{\rm{EE}}\label{OptA0}\\
	{\rm{s.t.}}\;\;&E_{k}\geq E_{k}^{\rm{min}},k\in\mathcal{K},\label{OptA1}\\
	&0\leq\beta_{k}\leq 1, k\in\mathcal{K},\label{OptA2}\\
		&\sum_{k=0}^{K}||{\bf{F}}{\bf{v}}_{k}||^2\leq P_{\rm{max}},\label{OptA3}
	\end{align}
\end{subequations}
where $(\rm{\ref{OptA1}})$ denotes the minimum requirement of the harvested energy for each user, and $(\rm{\ref{OptA3}})$ means the maximum transmit power constraint for the BS.  {One can observe that higher transmit power is needed for a larger $E_{k}^{\rm{min}}$. However, $(\rm{\ref{OptA3}})$ means that the total transmit power is limited. When $P_{\rm{max}}$ is small while $E_{k}^{\rm{min}}$ is large, problem $(\rm{\ref{OptA}})$ may be infeasible. Therefore, it cannot be guaranteed that  $(\rm{\ref{OptA}})$ is always feasible. In general, for a given $P_{\rm{max}}$, when $(\rm{\ref{OptA}})$ is infeasible, we have to decrease $E_{k}^{\rm{min}}$.}
\subsection{The Proposed Two-Loop Iterative Algorithm}

After obtaining the analog precoding $\bf{F}$, we define the equivalent channel as $\hat{{\bf{h}}}_k={\bf{h}}_k{\bf{F}}$. By introducing auxiliary variables $t_0$ and $t_k$, the original EE maximization problem can be equivalently expressed as
\begin{subequations}\label{OptB}
	\begin{align}
	\;\;\;\;&\underset{\left\{\{\beta_{k}\},\;\{{\bf{v}}_k\},\;t_0,\; \{t_k\}\right\}}{\rm{max}}\;\;\;\;\;\;\frac{\sum_{k=0}^{K}\log_2(1+t_k)}{P_{\rm{total}}}\label{OptB0}\\
	{\rm{s.t.}}\;\;&\frac{\beta_{k}|\hat{{\bf{h}}}_k{\bf{v}}_0|^2}{\beta_{k}\left(\sum_{i=1}^K|\hat{{\bf{h}}}_k{\bf{v}}_i|^2+\delta_0^2\right)+\delta_1^2}\geq t_0, k\in\mathcal{K},\label{OptB1}\\
	&\frac{\beta_{k}|\hat{{\bf{h}}}_k{\bf{v}}_k|^2}{\beta_{k}\left(\sum_{i\neq k}^K|\hat{{\bf{h}}}_k{\bf{v}}_i|^2+\delta_0^2\right)+\delta_1^2}\geq t_k, k\in\mathcal{K},\label{OptB2}\\
	&\varepsilon(1-\beta_{k})\left(\sum_{i=0}^K|\hat{{\bf{h}}}_k{\bf{v}}_i|^2+\delta_0^2\right)\geq E_{k}^{\rm{min}},k\in\mathcal{K},\label{OptB3}\\
	&0\leq\beta_{k}\leq 1, k\in\mathcal{K},\label{OptB4}\\
	&\sum_{k=0}^{K}||{\bf{F}}{\bf{v}}_k||^2\leq P_{\rm{max}}.\label{OptB5}
	\end{align}
\end{subequations}

It is evident that (\ref{OptB}) is a non-convex optimization problem due to  non-convex constraints (\ref{OptB1})-(\ref{OptB3}). To solve the above problem, we first equivalently transform the fractional objective function into the subtractive form. We denote $q^\star$ as the obtained maximum EE of the system, namely
\begin{eqnarray}
	q^\star\!=\!\frac{\sum_{k=0}^{K}\log_2(1+t_k^\star)}{\sum_{k=0}^{K}\xi||{\bf{F}}{\bf{v}}_k^\star||^2+P_{\rm{C}}}
	=\underset{\left\{\{\beta_{k}\},\;\{{\bf{v}}_k\},\;t_0,\; \{t_k\}\right\}}{\rm{max}}\frac{\sum_{k=0}^{K}\log_2(1+t_k)}{\sum_{k=0}^{K}\xi||{\bf{F}}{\bf{v}}_k||^2+P_{\rm{C}}},
\end{eqnarray}
 where $\left\{\{\beta_{k}\}, \{{\bf{v}}_k\}, t_0, \{t_k\}\right\}$ should satisfy constraints (\ref{OptB1})-(\ref{OptB5}). Then, we apply the following Theorem.
 \begin{theorem}\label{theorem1}
 	The maximum EE $q^{\star}$ is obtained if an only if 
 	\begin{eqnarray}
 	\begin{aligned}
 	\underset{\left\{\{\beta_{k}\},\;\{{\bf{v}}_k\},\;t_0,\; \{t_k\}\right\}}{\rm{max}}\;\;&{\sum_{k=0}^{K}\log_2(1+t_k)}-q^\star\left({\sum_{k=0}^{K}\xi||{\bf{F}}{\bf{v}}_k||^2+P_{\rm{C}}}\right)\\
 	=&{\sum_{k=0}^{K}\log_2(1+t_k^\star)}-q^\star\left({\sum_{k=0}^{K}\xi||{\bf{F}}{\bf{v}}_k^\star||^2+P_{\rm{C}}}\right)\\
 	=&0,
 	\end{aligned}
 	\end{eqnarray}
 	where ${\sum_{k=0}^{K}\log_2(1+t_k)}\geq 0$ and ${\sum_{k=0}^{K}\xi||{\bf{F}}{\bf{v}}_k||^2+P_{\rm{C}}}>0$.
 \end{theorem}
\begin{proof}
	The proof of Theorem~\ref{theorem1} can follow the similar approach to the one in~\cite{eeproof}.
\end{proof}

As a result, we need to solve the following optimization problem for a given $q$
\setlength{\mathindent}{0cm}
\begin{subequations}\label{OptC}
	\begin{align}
   &\underset{\left\{\{\beta_{k}\},\;\{{\bf{v}}_k\},\;t_0,\; \{t_k\}\right\}}{\rm{max}}\;\;{\sum_{k=0}^{K}\log_2(1+t_k)}-q\sum_{k=0}^{K}\xi||{\bf{F}}{\bf{v}}_k||^2\label{OptC0}\\
    {\rm{s.t.}}\;\;&{\rm{(\ref{OptB1})-(\ref{OptB5})}},
	\end{align}
\end{subequations}
where $q$ can be regarded as a parameter and we can denote the optimal value of (\ref{OptC0}) as $T(q)$. To this end, we have the following definition according to Theorem~\ref{theorem1}
\begin{eqnarray}
q=q^\star \Leftrightarrow T(q)=0,
\end{eqnarray} 
which means that searching the root for the nonlinear equation $T(q)=0$ is equivalent to solve (\ref{OptB}). It can be found that $T(q)$ is a strictly decreasing and convex function with respect to $q$, where  $T(q)>0$ with $q\rightarrow -\infty$ and $T(q)<0$ with $q\rightarrow \infty$. Therefore, we can use the classical Bi-section method to find  $T(q)=0$, which is summarized as Algorithm~\ref{algorithm3}. 

\begin{algorithm}[t]
	{\caption{The Bi-section-Based EE Resource Allocation Algorithm.}
		\label{algorithm3}
		{\bf{Initialize}} $q_{\rm{s}}=0$, $q_{\rm{b}}\gg 0$ with $T(q_{\rm{s}})>0$ and $T(q_{\rm{b}})<0$, a small constant $\epsilon$.\\
		\Repeat{$|T(q_{\rm{m}})|<\epsilon$}{Update $q_{\rm{m}}\leftarrow(q_{\rm{s}}+q_{\rm{b}})/2$,\\
			Solve problem (\ref{OptC}) and obtain $T(q_{\rm{m}})$,\\
			$q_{\rm{s}}\leftarrow q_{\rm{m}}$ if $T(q_{\rm{m}})>0$, else $q_{\rm{b}}\leftarrow q_{\rm{m}}$.}
		
	}
\end{algorithm} 

Apparently,  (\ref{OptC}) is still a non-convex optimization problem. Next, we define two new variables $\mu_k=1/\beta_k$, $\omega_k=1/(1-\beta_k)$ and reformulate the following optimization problem
\setlength{\mathindent}{0cm}
\begin{subequations}\label{OptD}
	\begin{align}
&\underset{\left\{\{\mu_{k}\},\;\{\omega_k\},\;\{{\bf{v}}_k\},\;t_0,\; \{t_k\}\right\}}{\rm{max}}\;{\sum_{k=0}^{K}\log_2(1\!+\!t_k)}-\sum_{k=0}^{K}\xi||{\bf{F}}{\bf{v}}_k||^2\label{OptD0}\\
	{\rm{s.t.}}&{|\hat{{\bf{h}}}_k{\bf{v}}_0|^2}\geq t_0\left(\sum_{i=1}^K|\hat{{\bf{h}}}_k{\bf{v}}_i|^2+\delta_0^2+\mu_k\delta_1^2\right) , k\in\mathcal{K},\label{OptD1}\\
	&|\hat{{\bf{h}}}_k{\bf{v}}_k|^2\geq t_k\left(\sum_{i\neq k}^K|\hat{{\bf{h}}}_k{\bf{v}}_i|^2+\delta_0^2+\mu_k\delta_1^2\right), k\in\mathcal{K},\label{OptD2}\\
	&\sum_{i=0}^K|\hat{{\bf{h}}}_k{\bf{v}}_i|^2+\delta_0^2\geq \frac{E_{k}^{\rm{min}}}{\varepsilon}\omega_k,k\in\mathcal{K},\label{OptD3}\\
	&\left\| \begin{array}{c}
	\mu_k-\omega_k \\
	2\end{array}\right\|_2 \leq \mu_k+\omega_k-2, k\in\mathcal{K},\label{OptD4}\\
	&{\rm{(\ref{OptB4}),\;(\ref{OptB5})}}.\label{OptD5}
	\end{align}
\end{subequations}

Meanwhile, we have the following theorem.
\begin{theorem}
	The optimal solution of (\ref{OptD}) is also the optimal solution of (\ref{OptC}). 
\end{theorem} 
\begin{proof}
   According to (\ref{OptD4}), we have,
   \begin{eqnarray}\label{Ax}
   	\mu_k+\omega_k\leq \mu_k\omega_k.
   \end{eqnarray}
  Due to $\mu_k\omega_k\geq 0$, we divide (\ref{Ax}) by $\mu_k\omega_k$ and have
     \begin{eqnarray}\label{Ab}
  \frac{1}{\mu_k}+\frac{1}{\omega_k}\leq 1.
  \end{eqnarray}

Assume $\left\{\{\mu_{k}^\star\},\;\{\omega_k^\star\},\;\{{\bf{v}}_k^\star\},\;t_0^\star,\; \{t_k^\star\}\right\}$ represents a global optimization solution of (\ref{OptD}), when (\ref{OptD4}) is satisfied with equality, namely $1/\mu_k^\star+1/\omega_k^\star=1$, it is clear that problems (\ref{OptC}) and (\ref{OptD}) are equivalent and we only need to replace $1/\mu_k^\star$ with $\beta_k^\star$ for $k\in \mathcal{K}$. Otherwise, if $1/\mu_k^\star+1/\omega_k^\star<1$, we scale $\left\{\{\mu_{k}^\star\},\;\{\omega_k^\star\}\right\}$ by $(1/\mu_k^\star+1/\omega_k^\star)$ and have 
\begin{eqnarray}
\begin{aligned}
&\frac{1}{\mu_k^\star(1/\mu_k^\star+1/\omega_k^\star)}+\frac{1}{\omega_k^\star(1/\mu_k^\star+1/\omega_k^\star)}\\
=&\frac{\mu_{k}^\star}{\mu_k^\star+\omega_k^\star}+\frac{\omega_k^\star}{\mu_k^\star+\omega_k^\star}=1.
\end{aligned}
\end{eqnarray}

It means that (\ref{OptD4}) can satisfy the equality and the harvested power does not violate the constraint (\ref{OptD3}) due to $1/\mu_k^\star+1/\omega_k^\star<1$. As a result, the values of $\sum_{i=1}^K|\hat{{\bf{h}}}_k{\bf{v}}_i|^2+\delta_0^2+\mu_k\delta_1^2$ and $\sum_{i\neq k}^K|\hat{{\bf{h}}}_k{\bf{v}}_k|^2+\delta_0^2+\mu_k\delta_1^2$ in (\ref{OptD1}) and (\ref{OptD2}) will become smaller. Meanwhile, we can obtain a large $\{t_0,t_k\}$ than the optimal $\{t_0^\star,t_k^\star\}$,  which means a greater rate can be obtained for the same ${\bf{v}}_k$. It is contradictory with our original assumption, and we finish the proof.  
\end{proof}

Next, we need to solve (\ref{OptD}), which is still a non-convex optimization problem due to the non-convex constraints (\ref{OptD1})-(\ref{OptD3}).  Let $\{\hat{{\bf{v}}}_{k}\} (k\in\mathcal{K})$ is a feasible solution and then, we define ${{\bf{v}}}_{k}=\hat{{\bf{v}}}_{k}+\Delta{{\bf{v}}}_{k}$ and have
\begin{eqnarray}\label{SCA1}
\begin{aligned}
|\hat{{\bf{h}}}_{k}{\bf{v}}_{k}|^2=&(\hat{{\bf{v}}}_{k}+\Delta{{\bf{v}}}_{k})^H{\hat{\bf{H}}}_{k}(\hat{{\bf{v}}}_{k}+\Delta{{\bf{v}}}_{k})\\
\geq& 2{\rm{Re}}\{(\hat{{{\bf{v}}}}_k)^H{\hat{\bf{H}}}_{k}\Delta{{\bf{v}}}_{k}\}+(\hat{{{\bf{v}}}}_k)^H{\hat{\bf{H}}}_{k}\hat{{\bf{v}}}_{k},
\end{aligned}
\end{eqnarray}
where ${\hat{\bf{H}}}_{k}={\hat{\bf{h}}}_{k}^H{\hat{\bf{h}}}_{k}$, $\Delta{{\bf{v}}}_{k}={{\bf{v}}}_{k}-\hat{{\bf{v}}}_{k}$. In this case, $|\hat{{\bf{h}}}_{k}{\bf{v}}_{k}|^2$ can be replaced by its convex approximations and the formulated problem can be solved iteratively. Accordingly, (\ref{OptD1})-(\ref{OptD3}) can be transformed~as
\begin{eqnarray}
 2{\rm{Re}}\{(\hat{{{\bf{v}}}}_0)^H{\hat{\bf{H}}}_{k}\Delta{{\bf{v}}}_{0}\}\!+\!(\hat{{{\bf{v}}}}_0)^H{\hat{\bf{H}}}_{k}\hat{{\bf{v}}}_{0}
 \!\geq\! t_0\left(\sum_{i=1}^K|\hat{{\bf{h}}}_k{\bf{v}}_i|^2\!+\!\delta_0^2\!+\!\mu_k\delta_1^2\right),\\
  2{\rm{Re}}\{(\hat{{{\bf{v}}}}_k)^H{\hat{\bf{H}}}_{k}\Delta{{\bf{v}}}_{k}\}\!+\!(\hat{{{\bf{v}}}}_k)^H{\hat{\bf{H}}}_{k}\hat{{\bf{v}}}_{k}
  \!\geq\! t_k\left(\sum_{i\neq k}^K|\hat{{\bf{h}}}_k{\bf{v}}_i|^2\!+\!\delta_0^2\!+\!\mu_k\delta_1^2\right),\\
  \sum_{i=0}^K2{\rm{Re}}\{(\hat{{{\bf{v}}}}_i)^H{\hat{\bf{H}}}_{k}\Delta{{\bf{v}}}_{i}\}+(\hat{{{\bf{v}}}}_i)^H{\hat{\bf{H}}}_{k}\hat{{\bf{v}}}_{i}+\delta_0^2 \geq {E_{k}^{\rm{min}}}\omega_k/{\varepsilon}.\label{sca1}
\end{eqnarray}

Then, we set the new variables $\tau_k\geq\sum_{i=1}^K|\hat{{\bf{h}}}_k{\bf{v}}_i|^2+\delta_0^2+\mu_k\delta_1^2$ and $\lambda_k\geq \sum_{i\neq k}^K|\hat{{\bf{h}}}_k{\bf{v}}_i|^2+\delta_0^2+\mu_k\delta_1^2$, and reformulate the following optimization as
\setlength{\mathindent}{0cm}
	\begin{subequations}\label{OptE}
	\begin{align}
	&\underset{\left\{\{\mu_{k}\},\;\{\omega_k\},\;\{\tau_k\},\;\{\lambda_k\},\;\{{\bf{v}}_k\},\;t_0,\; \{t_k\}\right\}}{\rm{max}}\;{\sum_{k=0}^{K}\log_2(1\!+\!t_k)}\!-\!q\sum_{k=0}^{K}\xi||{\bf{F}}{\bf{v}}_k||^2\label{OptE0}\\
	&{\rm{s.t.}}\;\; 2{\rm{Re}}\{(\hat{{{\bf{v}}}}_0)^H{\hat{\bf{H}}}_{k}\Delta{{\bf{v}}}_{0}\}+(\hat{{{\bf{v}}}}_0)^H{\hat{\bf{H}}}_{k}\hat{{\bf{v}}}_{0}\geq t_0\tau_k, k\in\mathcal{K},\label{OptE1}\\
&\;\;\;\;\;\sum_{i=1}^K|\hat{{\bf{h}}}_k{\bf{v}}_i|^2+\delta_0^2+\mu_k\delta_1^2\leq \tau_k,k\in\mathcal{K},\label{OptE2}\\
&\;\;\;\;\;2{\rm{Re}}\{(\hat{{{\bf{v}}}}_k)^H{\hat{\bf{H}}}_{k}\Delta{{\bf{v}}}_{k}\}+(\hat{{{\bf{v}}}}_k)^H{\hat{\bf{H}}}_{k}\hat{{\bf{v}}}_{k}\geq t_k\lambda_k, k\in\mathcal{K},\label{OptE3}\\
&\;\;\;\;\;\sum_{i\neq k}^K|\hat{{\bf{h}}}_k{\bf{v}}_i|^2+\delta_0^2+\mu_k\delta_1^2\leq \lambda_k,k\in\mathcal{K},\label{OptE4}\\
&\;\;\;\;\;{\rm{(\ref{OptB4}),\;(\ref{OptB5}),\;(\ref{OptD4}),\;(\ref{sca1})}}.
	\end{align}
\end{subequations}

\begin{algorithm}[t]
	{\caption{The Joint Digital Precoding and Power Splitting Ratio Iterative Algorithm.}
		\label{algorithm4}
		{\bf{Initialize}} $\{\tau_k^{[i-1]}\}, \{\lambda_k^{[i-1]}\}, \{{\bf{v}}_k^{[i-1]}\}, t_0^{[i-1]},\{t_k^{[i-1]}\}$, $i=1$, the maximum iteration times $T_{\rm{max}}$.\\
		\Repeat{$i=T_{\rm{max}}$}{
			Solve the optimization problem (\ref{OptF}) and obtain the optimal $\{\mu_{k}^{[i]}\},\;\{\omega_k^{[i]}\},\;\{\tau_k^{[i]}\},\;\{\lambda_k^{[i]}\},\;\{{\bf{v}}_k^{[i]}\},\;t_0,\; \{t_k^{[i]}\}$.\\
			Update $i\leftarrow i+1$.}
	}
\end{algorithm} 

So far, the only non-convex constraints are (\ref{OptE1}) and (\ref{OptE3}) in (\ref{OptE}). Similar to~\cite{SCA2}, we can obtain the upper of $t_0\tau_k$ and $t_k\lambda_k$ as
\begin{eqnarray}\label{sca4}
	\frac{t_0^{[i-1]}}{2\tau_k^{[i-1]}}\tau_k^2+\frac{\tau_k^{[i-1]}}{2t_0^{[i-1]}}t_0^2\geq t_0\tau_k,\;\;
\frac{t_k^{[i-1]}}{2\lambda_k^{[i-1]}}\lambda_k^2+\frac{\lambda_k^{[i-1]}}{2t_k^{[i-1]}}t_k^2\geq t_k\lambda_k,
\end{eqnarray}
where $\{t_0^{[i-1]},\{\tau_k^{[i-1]}\},\{t_k^{[i-1]}\},\{\lambda_k^{[i-1]}\}\}$ are the value of $\{t_0,\{\tau_k\},\{t_k\},\{\lambda_k\}\}$ at the $(i\!-\!1)$th iteration. Finally, we formulate the following optimization problem 
\setlength{\mathindent}{0cm}
\begin{subequations}\label{OptF}
	\begin{align}
	&\underset{\left\{\{\mu_{k}\},\;\{\omega_k\},\;\{\tau_k\},\;\{\lambda_k\},\;\{{\bf{v}}_k\},\;t_0,\; \{t_k\}\right\}}{\rm{max}}\;{\sum_{k=0}^{K}\log_2(1\!+\!t_k)}\!-\!q\sum_{k=0}^{K}\xi||{\bf{F}}{\bf{v}}_k||^2\label{OptF0}\\
	&{\rm{s.t.}}\;\;2{\rm{Re}}\{(\hat{{{\bf{v}}}}_0)^H{\hat{\bf{H}}}_{k}\Delta{{\bf{v}}}_{0}\}\!+\!(\hat{{{\bf{v}}}}_0)^H{\hat{\bf{H}}}_{k}\hat{{\bf{v}}}_{0}\!\geq\! \frac{t_0^{[i-1]}}{2\tau_k^{[i-1]}}\tau_k^2\!+\!\frac{\tau_k^{[i-1]}}{2t_0^{[i-1]}}t_0^2,\label{OptF1}\\
	&\;\;\;\;\;\;2{\rm{Re}}\{(\hat{{{\bf{v}}}}_k)^H{\hat{\bf{H}}}_{k}\Delta{{\bf{v}}}_{k}\}\!+\!(\hat{{{\bf{v}}}}_k)^H{\hat{\bf{H}}}_{k}\hat{{\bf{v}}}_{k}\!\geq\! \frac{t_k^{[i-1]}}{2\lambda_k^{[i-1]}}\lambda_k^2\!+\!\frac{\lambda_k^{[i-1]}}{2t_k^{[i-1]}}t_k^2,\label{OptF3}\\
	&\;\;\;\;\;\;{\rm{(\ref{OptB4}),\;(\ref{OptB5}),\;(\ref{OptD4}),\;(\ref{sca1}),\;(\ref{OptE2}),\;(\ref{OptE4})}}.\label{OptF4}
	\end{align}
\end{subequations}

So far, it is clear that (\ref{OptF}) is a convex optimization problem, which can be solved by standard convex optimization technique, e.g., interior-point method or CVX tool box. Summarily, solving the original problem (\ref{OptC}), we need to iteratively solve the optimal values of $\left\{\{\mu_{k}\}, \{\omega_k\}, \{\tau_k\}, \{\lambda_k\}, \{{\bf{v}}_k\}, t_0,\{t_k\}\right\}$ via (\ref{OptF}). In addition, since obtained $\left\{\{\mu_{k}^{[i]}\}, \{\omega_k^{[i]}\}, \{\tau_k^{[i]}\}, \{\lambda_k^{[i]}\}, \{{\bf{v}}_k^{[i]}\}, t_0^{[i]},\{t_k^{[i]}\}\right\}$ are the optimal solutions at the $i$th iteration, iteratively updating these variables will increase or maintain the value of the objective function (\ref{OptF0}). To this end, the proposed iterative algorithm will converge to at least a local optimal solution, which is summarized in Algorithm~\ref{algorithm4}.

Next, we analyze the complexity of Algorithm~\ref{algorithm4}. In fact, (\ref{OptF}) is a second-order cone optimization problem, and the worst-case complexity of solving  (\ref{OptF}) with second-order cone form is $\mathcal{O}([KN_{\rm{TX}}+5K+1]^{3.5})$~\cite{cone}, where $KN_{\rm{TX}}+5K+1$ denotes the number of variables.

\section{ZF-Based Low-complexity Algorithm}
To reduce the complexity for solving (\ref{OptC}), we propose a ZF-based low complexity algorithm. Specifically,  we first apply the  ZF technique to cancel the inter-unicast interference. Let ${\bf{H}}=[\hat{{\bf{h}}}_1^T,\hat{{\bf{h}}}_2^T,\dots,\hat{{\bf{h}}}_K^T]^T$, which includes the equivalent downlink channel from the BS to all $K$ users. Then, the precoding matrix can be written as ${\bf{V}}={\bf{H}}^H({\bf{H}}{\bf{H}}^H)^{-1}$, and the digital precoding for unicast  of the $k$th user can be expressed as ${\bf{v}}_k={\bf{V}}_k/\|{\bf{F}}{\bf{V}}_k\|$, where ${\bf{V}}_k$ denotes the $k$th row of ${\bf{V}}$.

After applying ZF precoding, the multi-user interference can be removed from (\ref{sinr1}) and (\ref{sinr2}), and the user's SINR can be rewritten as
\begin{eqnarray}\label{sinr11}
\gamma_k^0=\frac{\beta_{k}|\hat{{\bf{h}}}_k{\bf{v}}_0|^2}{\beta_{k}\left(p_k|\hat{{\bf{h}}}_k{\bf{v}}_k|^2+\delta_0^2\right)+\delta_1^2}, k\in\mathcal{K},
\end{eqnarray} 
and
\begin{eqnarray}\label{sinr21}
\gamma_k=\frac{\beta_{k}p_k|\hat{{\bf{h}}}_k{\bf{v}}_k|^2}{\beta_{k}\delta_0^2+\delta_1^2}, k\in\mathcal{K},
\end{eqnarray}
where $p_k$ stands for the unicast transmit power for the $k$th user. In this case, we only need to optimize the transmit power $\{p_k\}$ and multicast precoding ${\bf{v}}_0$. Next, we define two new variables  $g_k$ and $o_k$, where $g_k\geq \frac{1}{\beta_k}$ and $o_k\geq \frac{E_{k}^{\rm{min}}}{\varepsilon(1-\beta_k)}$. As a result, the optimization problem (\ref{OptC}) can be rewritten as
\setlength{\mathindent}{0cm}
\begin{subequations}\label{OptH}
	\begin{align}
	&\underset{\left\{\{\beta_{k}\},\;\{{{p}}_k\},{\bf{v}}_0\;t_0,\; \{t_k\}\right\}}{\rm{max}}\;\;{\sum_{k=0}^{K}\log_2(1\!+\!t_k)}\!-\!q\xi(\sum_{k=1}^{K}p_k+\|{\bf{F}}{\bf{v}}_0\|^2)\label{OptH0}\\
	{\rm{s.t.}}\;\;&\frac{|\hat{{\bf{h}}}_k{\bf{v}}_0|^2}{p_k|\hat{{\bf{h}}}_k{\bf{v}}_k|^2+\delta_0^2+g_k\delta_1^2}\geq t_0, k\in\mathcal{K},\label{OptH1}\\
	&\frac{p_k|\hat{{\bf{h}}}_k{\bf{v}}_k|^2}{\delta_0^2+g_k\delta_1^2}\geq t_k, k\in\mathcal{K},\label{OptH2}\\
	&p_k|\hat{{\bf{h}}}_k{\bf{v}}_k|^2+|\hat{{\bf{h}}}_k{\bf{v}}_0|^2+\delta_0^2\geq o_{k},k\in\mathcal{K},\label{OptH3}\\
	&g_k\geq \frac{1}{\beta_k}, k\in\mathcal{K},\label{OptH4}\\
	&o_k\geq \frac{E_{k}^{\rm{min}}}{\varepsilon(1-\beta_k)}, k\in\mathcal{K},\label{OptH5}\\
	&\|{\bf{F}}{\bf{v}}_0\|^2+\sum_{k=1}^Kp_k\leq P_{\rm{max}},\label{OptH6}\\
	&{\rm{(\ref{OptB4})}}.
	\end{align}
\end{subequations}

One can observe that (\ref{OptH1})-(\ref{OptH5}) are all non-convex constraints. Next, we will apply some approximation techniques to transform them into convex ones. 

Firstly, combining with (\ref{SCA1}), (\ref{OptH3}) can  be  transformed into the convex constraint as
\begin{eqnarray}\label{con41}
p_k|\hat{{\bf{h}}}_k{\bf{v}}_k|^2+2{\rm{Re}}\{(\hat{{{\bf{v}}}}_0)^H{\hat{\bf{H}}}_{k}\Delta{{\bf{v}}}_{0}\}+(\hat{{{\bf{v}}}}_0)^H{\hat{\bf{H}}}_{k}\hat{{\bf{v}}}_{0}+\delta_0^2\geq o_k, k\in\mathcal{K}.
\end{eqnarray}

To deal with the non-convex constraints (\ref{OptH4}) and (\ref{OptH5}), we apply the Schur complement lemma~\cite{RefDai} and transform them into the following convex matrix form  constraints,
\begin{eqnarray}\label{Con1}
\left[ \begin{array}{ccc}
g_{k} & 1 \\
1 & \beta_{k} 
\end{array} 
\right ]\geq {\bf{0}}, k\in\mathcal{K},
\end{eqnarray} 
and
\begin{eqnarray}\label{Con2}
\left[ \begin{array}{ccc}
o_{k} & \sqrt{E_{k}^{\rm{min}}}\\
\sqrt{E_{k}^{\rm{min}}} & \varepsilon(1-\beta_{k}) 
\end{array} 
\right ]\geq {\bf{0}}, k\in\mathcal{K}.
\end{eqnarray} 

\begin{algorithm}[t]
	{\caption{The ZF-Based Low-Complexity Algorithm.}
		\label{algorithm5}
		{\bf{Initialize}} $\{g_k^{[i-1]}\}, \{t_k^{[i-1]}\}, \{{\bf{v}}_0^{[i-1]}\}$, $i=1$, the maximum iteration times $T_{\rm{max}}$.\\
		\Repeat{$i=T_{\rm{max}}$}{
			Solve the optimization problem (\ref{OptG}) and obtain the optimal $\{\beta_k^{[i]}\}, \{p_k^{[i]}\}, \{g_k^{[i]}\}, \{{\bf{v}}_0^{[i]}\}, t_0^{[i]},\{t_k^{[i]}\}$.
			Update $i\leftarrow i+1$.\\}
	}
\end{algorithm}

In addition, we find that the left side of (\ref{OptH1}) is a quadratic-over-affine function, which is jointly convex with respect to $\{{\bf{v}}_0,p_k,g_k\}$. Based on this, we define $f({\bf{v}}_0,p_k,g_k)=\frac{|\hat{{\bf{h}}}_k{\bf{v}}_0|^2}{p_k|\hat{{\bf{h}}}_k{\bf{v}}_k|^2+\delta_0^2+g_k\delta_1^2}, k\in\mathcal{K}$ and the first order Taylor series expansion of  $f({\bf{v}}_0,p_k,g_k)$ can be written as
\begin{eqnarray}
\begin{aligned}
	f({\bf{v}}_0,p_k,g_k)\approx\frac{2\left(\hat{{\bf{h}}}_k{\bf{v}}_0^{[i-1]}\right)^\ast}{\Gamma_k^{[i-1]}}\hat{{\bf{h}}}_k{\bf{v}}_0
	-\frac{|\hat{{\bf{h}}}_k{\bf{v}}_0^{[i-1]}|^2}{(\Gamma_k^{[i-1]})^2}\Gamma_k,
\end{aligned}
\end{eqnarray}
where $\Gamma_k=p_k|\hat{{\bf{h}}}_k{\bf{v}}_k|^2+\delta_0^2+g_k\delta_1^2$ and $\Gamma_k^{[i-1]}=p_k^{[i-1]}|\hat{{\bf{h}}}_k{\bf{v}}_k|^2+\delta_0^2+g_k^{[i-1]}\delta_1^2$. In addition, ${\bf{v}}_0^{[i-1]}$, $p_k^{[i-1]}$ and $g_k^{[i-1]}$ denote the value of ${\bf{v}}_0$, $p_k$ and $g_k$ at the $[i-1]$th iteration, respectively. As a result, (\ref{OptH1}) can be written as the following convex constraint
\begin{eqnarray}\label{con46}
	\frac{2\left(\hat{{\bf{h}}}_k{\bf{v}}_0^{[i-1]}\right)^\ast}{\Gamma_k^{[i-1]}}\hat{{\bf{h}}}_k{\bf{v}}_0
	-\frac{|\hat{{\bf{h}}}_k{\bf{v}}_0^{[i-1]}|^2}{(\Gamma_k^{[i-1]})^2}\Gamma_k\geq t_0, k\in\mathcal{K}.
\end{eqnarray}

Finally, (\ref{OptH2}) can be expressed as
 \begin{eqnarray}\label{hao3}
 \frac{|\hat{{\bf{h}}}_k{\bf{v}}_k|^2}{\delta_1^2}p_k-\frac{\delta_0^2}{\delta_1^2}t_k\geq t_kg_k, k\in\mathcal{K}.
 \end{eqnarray}
 Similar to (\ref{sca4}), we can obtain the upper bound of $t_kg_k$ as
 \begin{eqnarray}\label{sca5}
 \frac{t_k^{[i-1]}}{2g_k^{[i-1]}}g_k^2+\frac{g_k^{[i-1]}}{2t_k^{[i-1]}}t_k^2\geq t_kg_k,
 \end{eqnarray}
 where $t_k^{[i-1]}$ and $g_k^{[i-1]}$, respectively, are the value of $t_k$ and $g_k$ at the $[i-1]$th iteration.  After that, $(\ref{OptH2})$ can be transformed into the following convex constraint
  \begin{eqnarray}\label{hao4}
 \frac{|\hat{{\bf{h}}}_k{\bf{v}}_k|^2}{\delta_1^2}p_k-\frac{\delta_0^2}{\delta_1^2}t_k\geq  \frac{t_k^{[i-1]}}{2g_k^{[i-1]}}g_k^2+\frac{g_k^{[i-1]}}{2t_k^{[i-1]}}t_k^2, k\in\mathcal{K}.
 \end{eqnarray}
Accordingly, we need to iteratively solve the following convex optimization problem.
\setlength{\mathindent}{0cm}
\begin{subequations}\label{OptG}
	\begin{align}
	&\underset{\left\{\{\beta_{k}\},\;\{{{p}}_k\},\;\{{{g}}_k\}, {\bf{v}}_0\;t_0,\; \{t_k\}\right\}}{\rm{max}}\;\;{\sum_{k=0}^{K}\log_2(1\!+\!t_k)}\!-\!q\xi(\sum_{k=1}^{K}p_k+\|{\bf{F}}{\bf{v}}_0\|^2)\label{OptG0}\\
	&{\rm{s.t.}}\;\;
	{\rm{(\ref{OptB4}), (\ref{con41}), (\ref{Con1}), (\ref{Con2}), (\ref{con46}), (\ref{hao4})}}.
	\end{align}
\end{subequations}

\begin{figure}[t]
	\begin{center}
		\includegraphics[width=9cm,height=6.5cm]{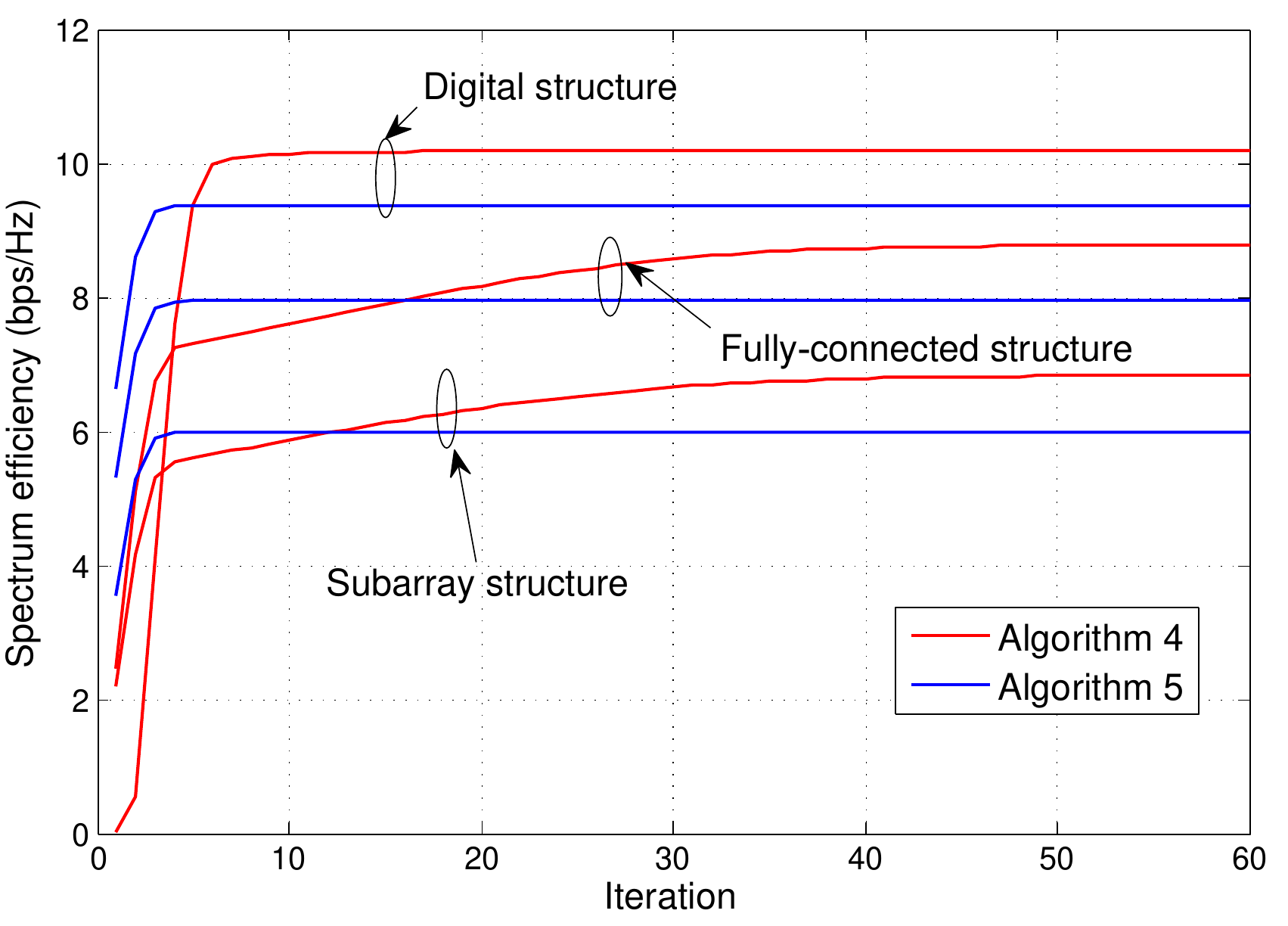}
		\caption{Spectrum efficiency versus iteration.}
		\label{fig3}
	\end{center}
\end{figure}
\begin{figure}[t]
	\begin{center}
		\includegraphics[width=9cm,height=6.5cm]{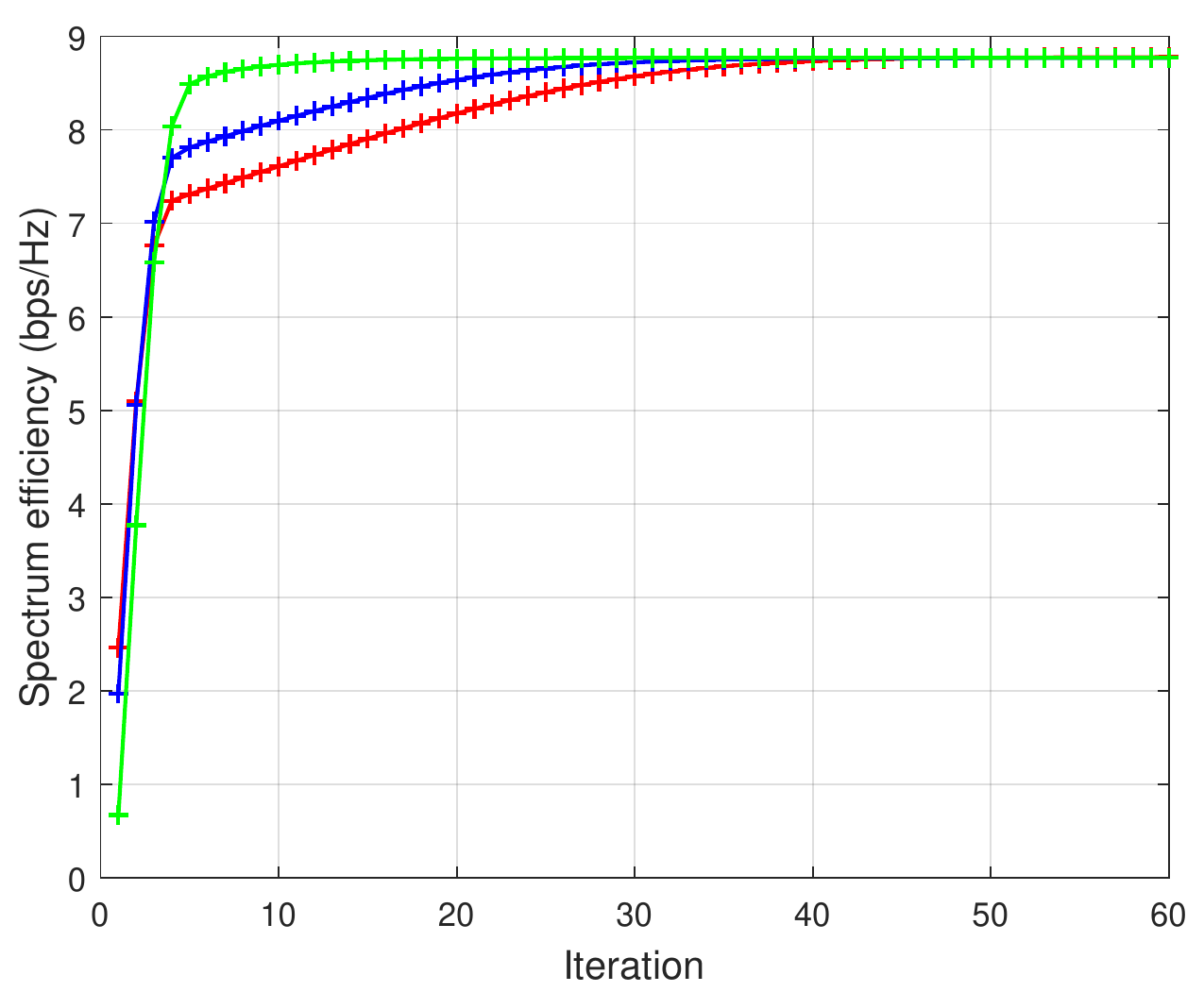}
		\caption{Energy efficiency versus iteration under different initial values.}
		\label{fig4A}
	\end{center}
\end{figure}

We can solve the above problem by convex solvers (e.g., CVX), which is summarized as Algorithm~\ref{algorithm5}. Similar to (\ref{OptF}), the worst-case complexity of solving  (\ref{OptG}) with second-order cone form is $\mathcal{O}([N_{\rm{TX}}+4K+1]^{3.5})<\mathcal{O}([KN_{\rm{TX}}+5K+1]^{3.5})$, where $N_{\rm{TX}}+4K+1$ denotes the number of variables in~(\ref{OptG}).

{Summarily, compared with the  unicast beamforming design, the significant challenges for joint consideration of multicast and unicast includes the following two aspects: $(\it{a})$ The energy efficiency (EE) is defined as
$\eta_{\rm{EE}}\!=\!\frac{\underset{\forall k}{\min}\{\log_2(1\!+\!\gamma_k^0)\}\!+\!\sum_{k=1}^{K}\log_2(1\!+\!\gamma_k)}{P_{\rm{total}}},$
where $\underset{\forall k}{\min}\{\log_2(1\!+\!\gamma_k^0)\}$ denotes the multicast rate. Then, we transform $\eta_{\rm{EE}}$  into a tractable one given by $\frac{\sum_{k=0}^{K}\log_2(1+t_k)}{P_{\rm{total}}}$ via bringing auxiliary variables $t_0$ and $t_k$, and then advanced convex approximated techniques are applied and an iterative algorithm is proposed, i.e, Algorithm 4. 
$(\it{b})$ For joint multicast-unicast transmission, when ZF technique is used to cancel the inter-unicast interference, we still need to design multicast beamforming such as our proposed Algorithm~5, which is  challenging and different from the unicast beamforming design.}

\section{Simulation Results}
In this section, simulation results are provided to illustrate the effectiveness of the proposed algorithms. We assume that the BS has a coverage of 30 meters, and the path loss is modeled as $69.4+24\log_{10}(D)$ dB, where $D$ denotes the distance in meter. We assume that there are 8 paths for the mmWave channel, and the azimuth angle of departure at BS is uniformly distributed over $[0, 2\pi]$. The BS is equipped with $N_{\rm{TX}}=256$ antennas and $N_{\rm{RF}}=4$ RF chains, where we set $d=\lambda/2$. The noise power, $\delta_0^2$ and $\delta_1^2$, are set -80 dBm and -60 dBm, respectively. The energy conversion efficiency $\varepsilon$  and the inefficiency of the power amplifier $\xi$ are set as 0.5 and 0.38, respectively. In addition, we set $P_{\rm{BB}}\!=\!200$ mW, $P_{\rm{RF}}\!=\!300$ mW, $P_{\rm{PS}}\!=\!40$ mW~\cite{Ref26a}. Meanwhile, the minimum harvested energy is $E_k^{\rm{min}}=100$ $\mu$W for all users~\cite{Ref26a},~\cite{Gen}, the number of users is set as $K=2$.

{Fig.~\ref{fig3} shows the convergence performance of the proposed two algorithms (e.g., Algorithms~\ref{algorithm4} and~\ref{algorithm5}) under different antenna structures, including digital structure (e.g., each antenna is connected to a dedicated RF chain), fully-connected structure and subarray structure, where we set $q=0$ and $P_{\rm{max}}=30$~dBm. It is clear that although the SE under Algorithm~\ref{algorithm4} is higher than that under Algorithm~\ref{algorithm5}, Algorithm~\ref{algorithm5} can converge speedily after 5 iterations, while Algorithm~\ref{algorithm4} needs about 50 iterations. Therefore, the ZF-based method can speedily obtain the solutions of the problem with a small performance loss. In addition, it can also be found that the SE under digital structure is the highest compared to another two structures, but its energy consumption and hardware complexity are high.} {In addition, Fig.~\ref{fig4A} shows the influence of different initial values on its solutions, where we consider the Algorithm~5 and fully-connected structure. From this figure, we can obtain that the proposed algorithm always converges to the same point under different initial values. However, initialization does slightly influence the convergence speed.}

\begin{figure}[t]
	\begin{center}
		\includegraphics[width=9cm,height=6.5cm]{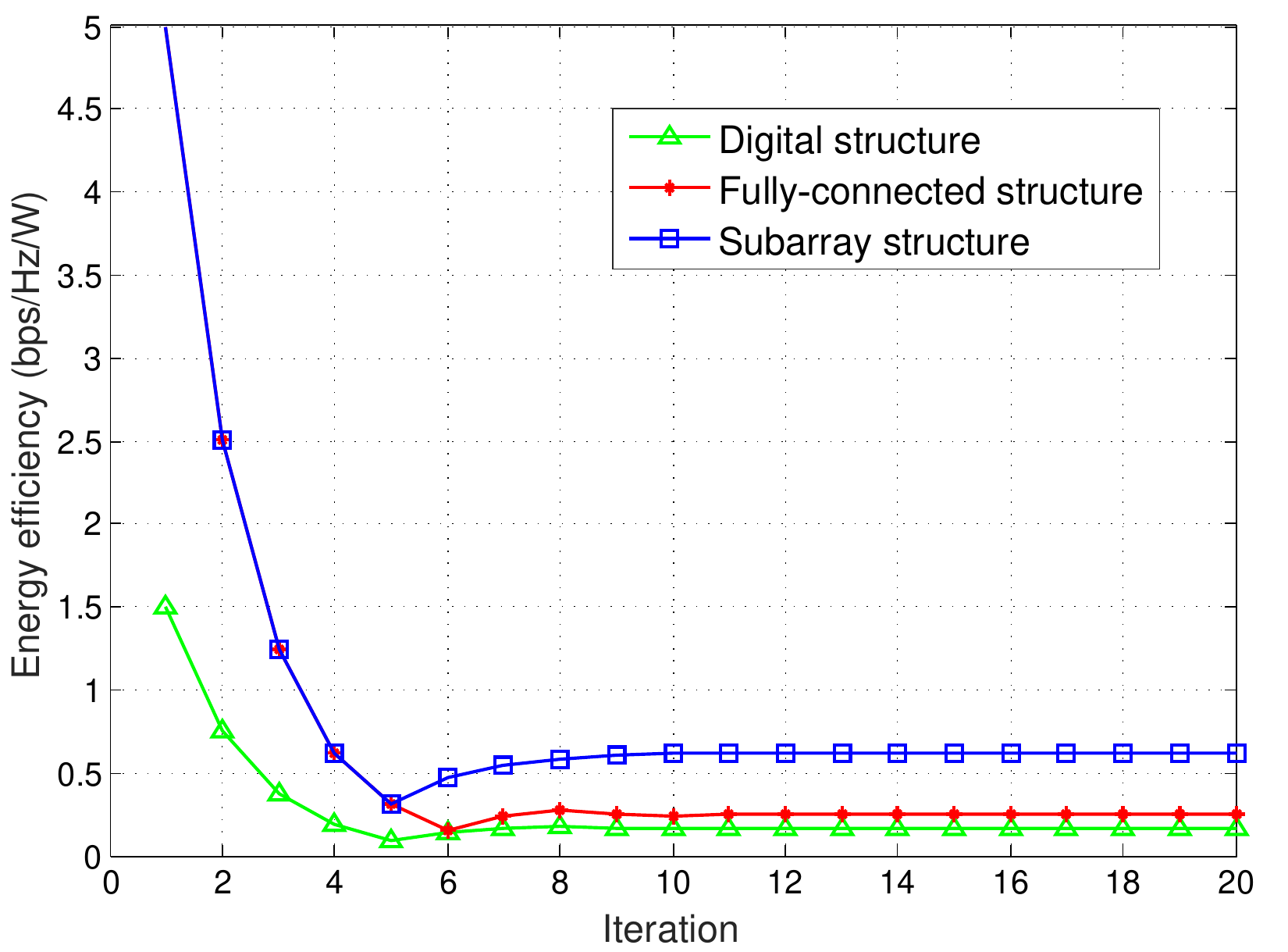}
		\caption{Energy efficiency versus iteration.}
		\label{fig4}
	\end{center}
\end{figure}
\begin{figure}[t]
	\begin{center}
		\includegraphics[width=9cm,height=6.5cm]{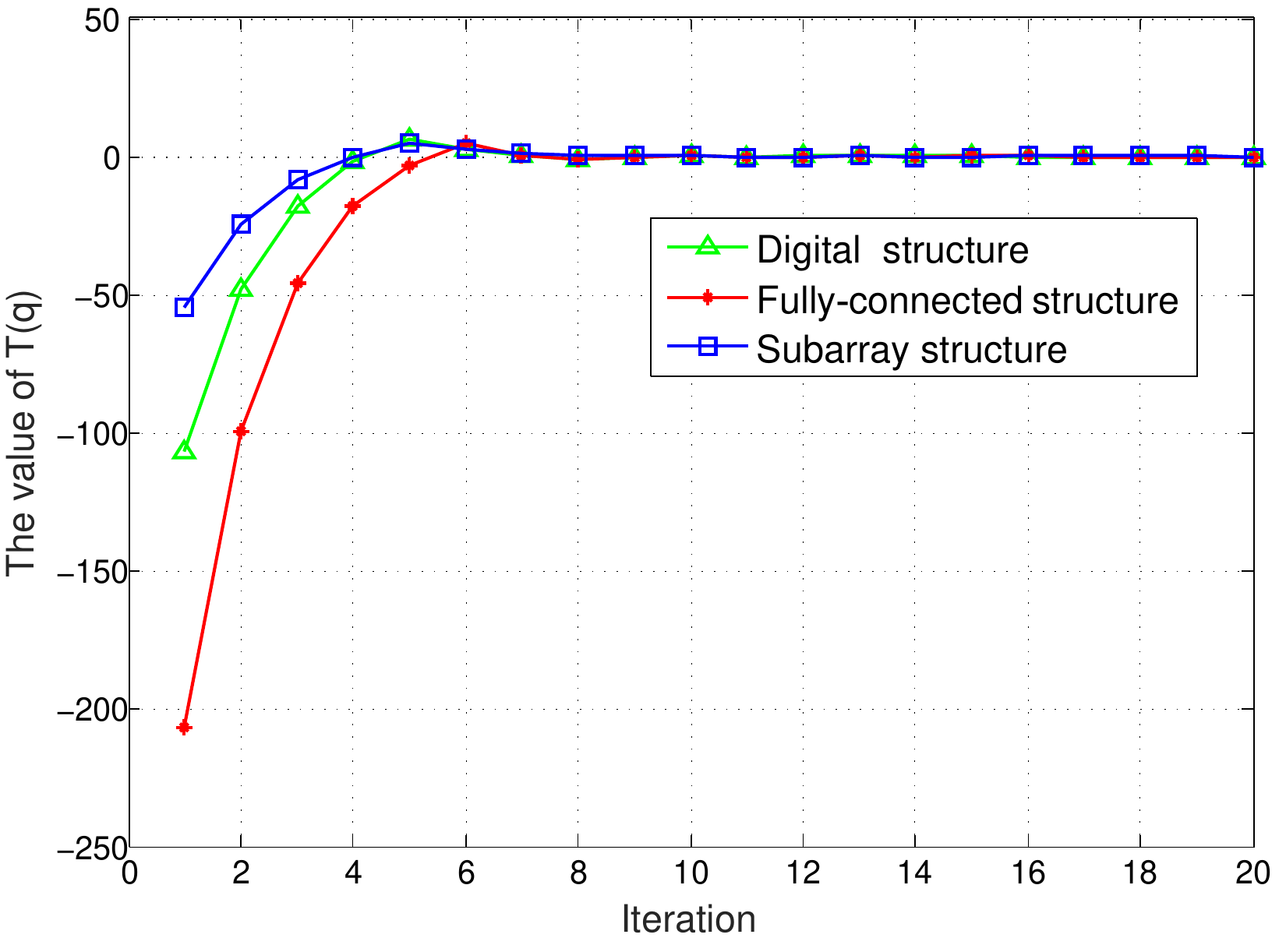}
		\caption{The value of $T(q)$ versus iteration.}
		\label{fig5}
	\end{center}
\end{figure}
\begin{figure}[t]
	\begin{center}
		\includegraphics[width=9cm,height=6.5cm]{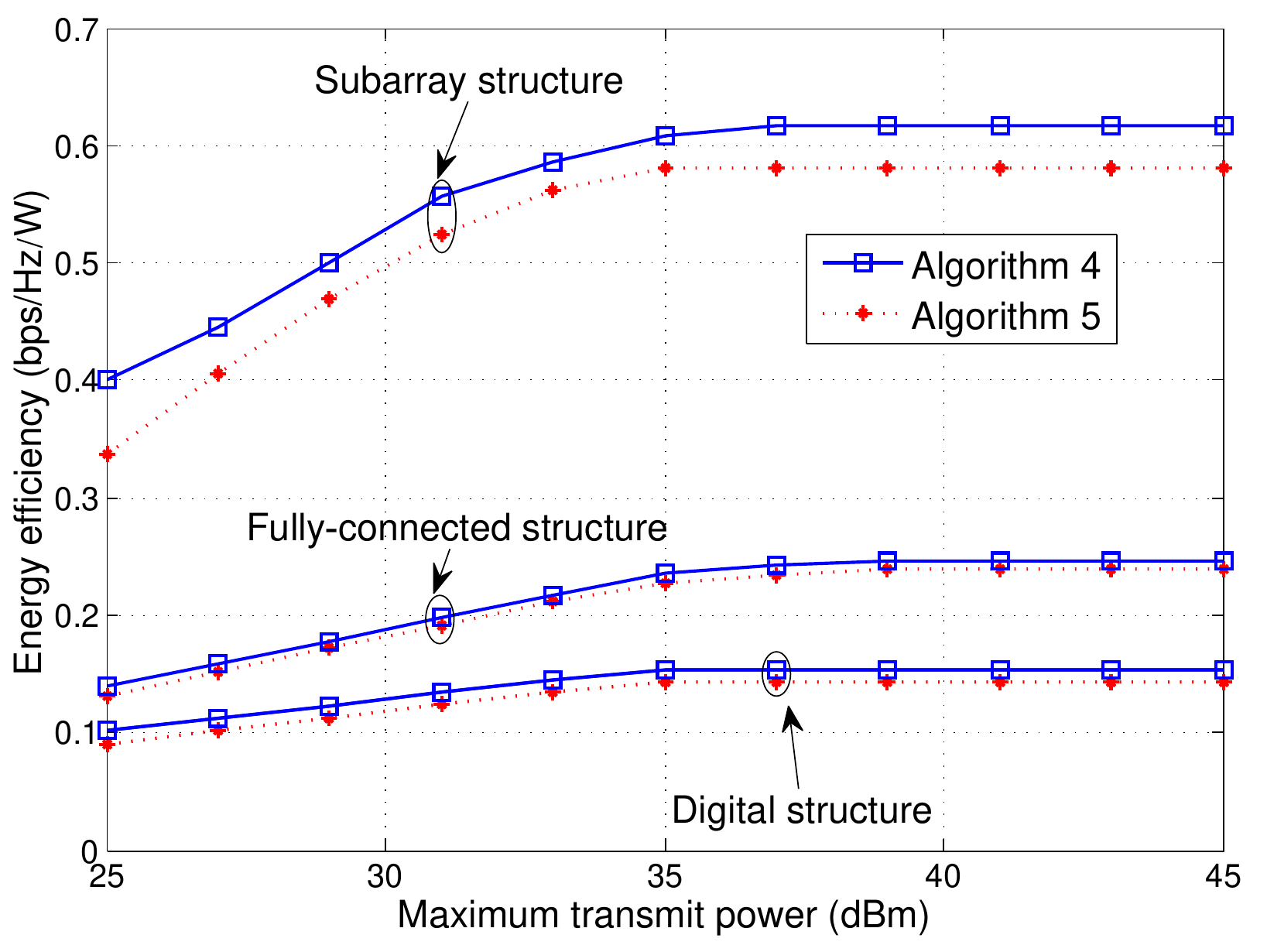}
		\caption{Energy efficiency versus maximum transmit power at the BS.}
		\label{fig6}
	\end{center}
\end{figure}
\begin{figure}[t]
	\begin{center}
		\includegraphics[width=9cm,height=6.5cm]{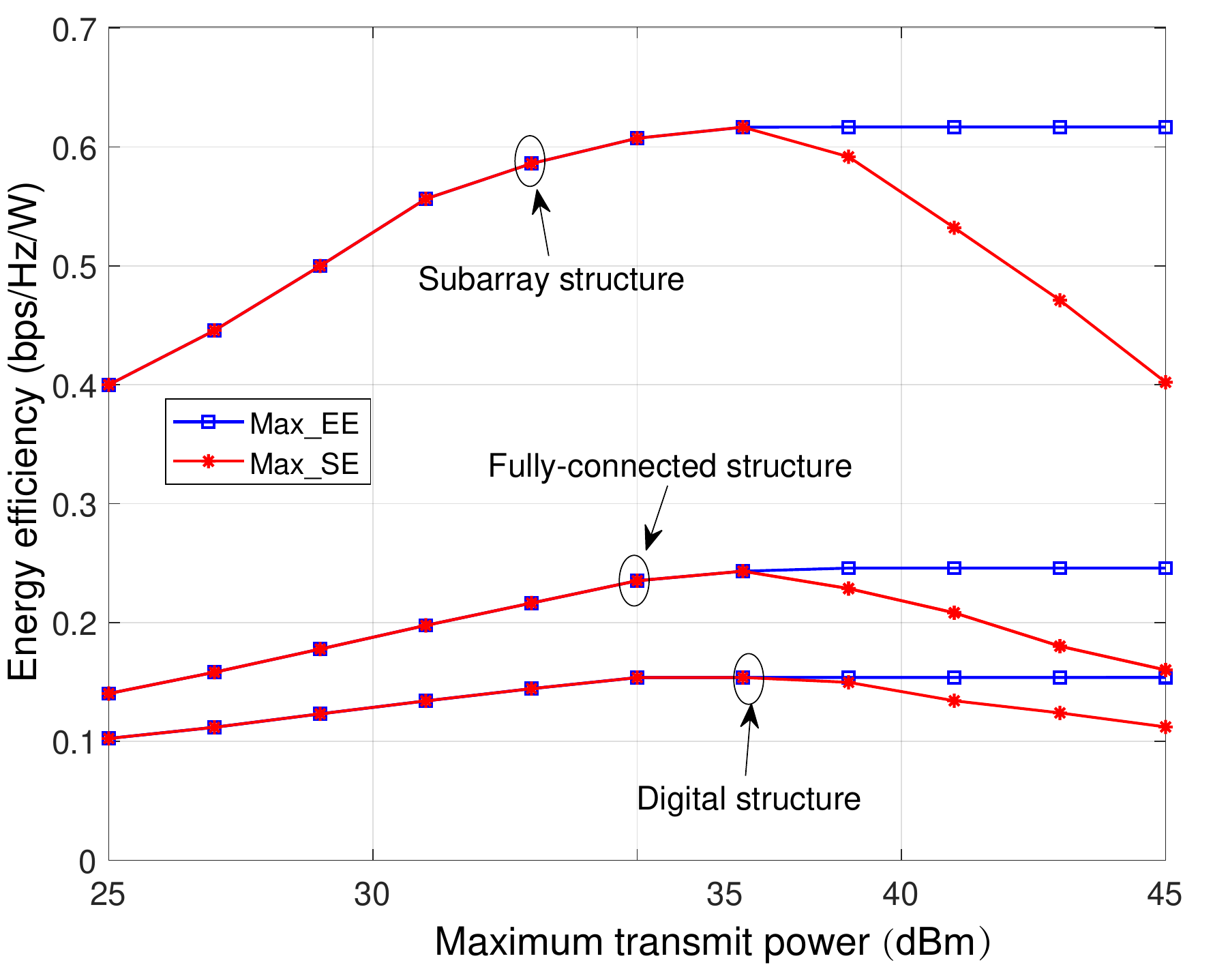}
		\caption{Energy efficiency versus maximum transmit power at the BS under different optimization schemes.}
		\label{fig8}
	\end{center}
\end{figure}

Figs.~\ref{fig4} and~\ref{fig5} show the convergence property of Algorithm~\ref{algorithm3} under different antenna structures, and we set $P_{\rm{max}}=40$ dBm. We solve problem~(\ref{OptC}) via   Algorithm~\ref{algorithm4}.  From Fig.~\ref{fig4}, one can observe that the EE trends to converge after 8 iterations. In addition, one can see that the EE under subarray structure is higher than that under another two structures. This is because its circuit power consumption is low due to the small number of RF chains and phase shifters.  In addition,  the value of $T(q)$ must be zero according to Theorem~\ref{theorem1}, and Fig.~\ref{fig4} also verifies this point.

{Fig.~\ref{fig6} shows the EE versus maximum transmit power at the BS with $N_{\rm{RF}}=4$. We can observe that the EE first increases and then saturates as $P_{\rm{max}}$ increases. It is understandable that a larger transmit power can obtain a higher SE, but the improved ratio will be lower and lower as transmit power increases. Therefore, the EE will reach the point of diminishing returns when the transmit power continues to increase. In addition, it is clear that the EE under subarray structure is the highest and under digital structure is the lowest due to the huge power consumption of RF chains.}

We examine the EE of the system under different  optimization schemes in Fig.~\ref{fig8}. ``Max\_EE" stands that the EE of the system when the EE is maximized, while ``Max\_SE" represents the EE of the system  when the SE is maximized. In addition, Algorithm~\ref{algorithm4} is used to solve the optimization problem. It can be observed that when the maximum transmit power is identical, the EE under two optimizations schemes is the same. As the maximum transmit power increases,  the EE reaches maximum and remains constant under ``Max\_EE" scheme, while the EE decreases under ``Max\_SE" scheme. In fact, the objection of ``Max\_SE" scheme is to maximize the SE without considering the power consumption. As a result, the EE may decrease for larger transmit power.

{Fig.~\ref{fig8A} shows the EE versus the minimum harvested energy, and we set $P_{\rm{max}}=45$ dBm. One can observe that the EE keeps a constant when $E_k^{\rm{min}}$ is relatively small, e.g., $E_k^{\rm{min}}\in[0.1\;\;0.4]$. This is because there are  redundant power at the BS that can be used to satisfy the requirement of the harvested energy. However, when $E_k^{\rm{min}}$ is large, more power has to be used to transform into the energy, and thus the EE is decreased.  Based on this, it means that higher EH may lead to a lower EE. Therefore, in general, it is unavailable to increase EE and EH simultaneously, and we have to sacrifice one for improving the other.}

 \begin{figure}[t]
 	\begin{center}
 		\includegraphics[width=9cm,height=6.5cm]{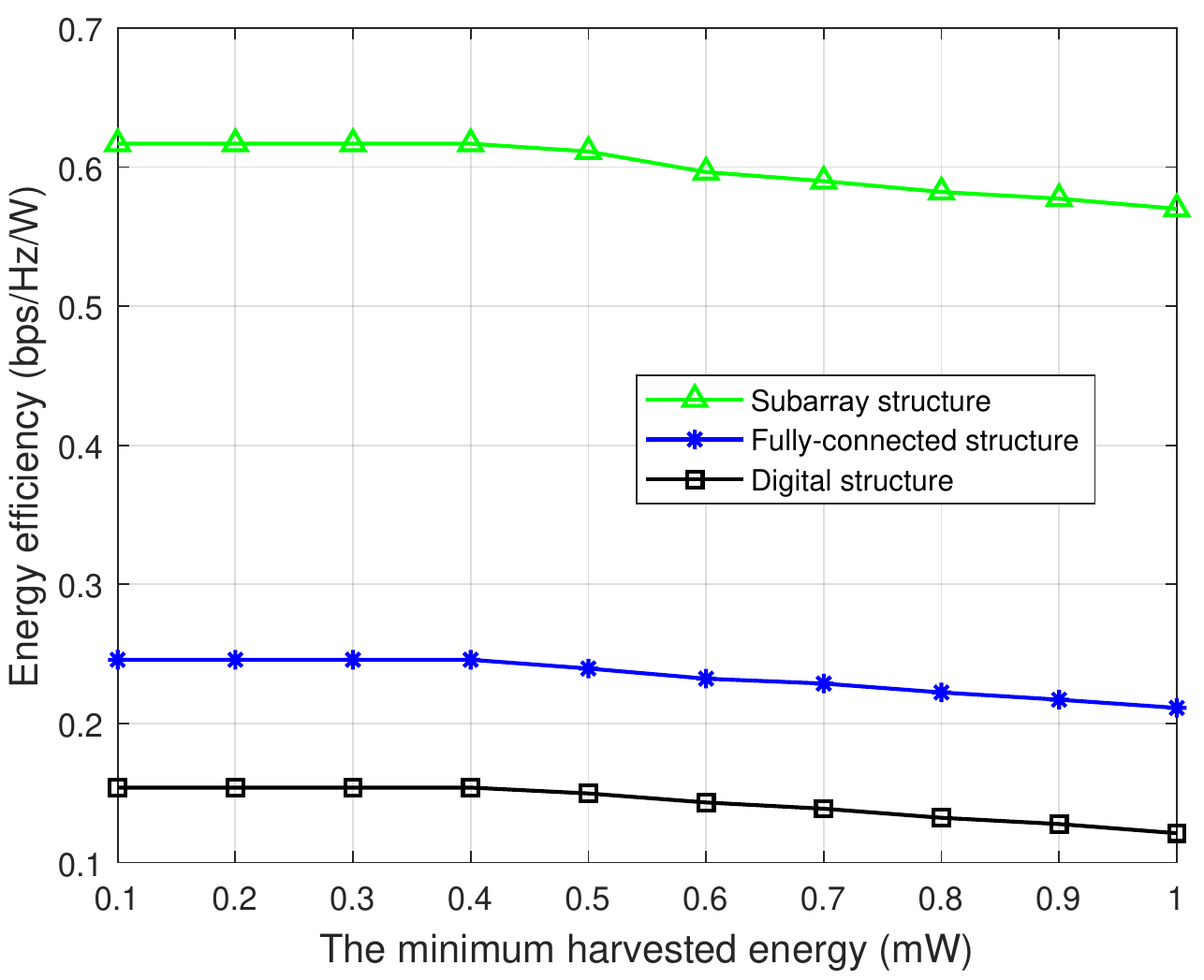}
 		\caption{Energy efficiency versus minimum harvested energy.}
 		\label{fig8A}
 	\end{center}
 \end{figure}

Finally, the tradeoff between the EE and SE is demonstrated under the subarray structure in Fig.~\ref{fig9}. It can be observed that  the EE increases with the SE when the SE is small. For a larger SE, the EE will decreases, which means that a large SE does not lead to a higher EE, and vice versa. Therefore, there exists a tradeoff between the EE and SE, specifically for a higher SE.   
\section{Conclusions}
In this paper, we investigated the EE maximization problem in a joint multicast-unicast mmWave communication system with SWIPT. We first designed the analog precoding for two sparse RF chain structures. Next, we proposed a two-loop algorithm to solve the formulated EE optimization problem. The Bi-section algorithm is adopted in outer loop. Subsequently, we developed two iterative algorithms (e.g., Algorithms~\ref{algorithm4} and \ref{algorithm5}) for the inner loop. Simulation results showed that although the performance of Algorithm~\ref{algorithm5} is slightly inferior to  that of Algorithm~\ref{algorithm4}, fast convergence can be achieved. In addition, it can also be  observed that there still  a tradeoff between EE and SE, specifically for a larger SE.

 \begin{figure}[t]
	\begin{center}
		\includegraphics[width=9cm,height=6.5cm]{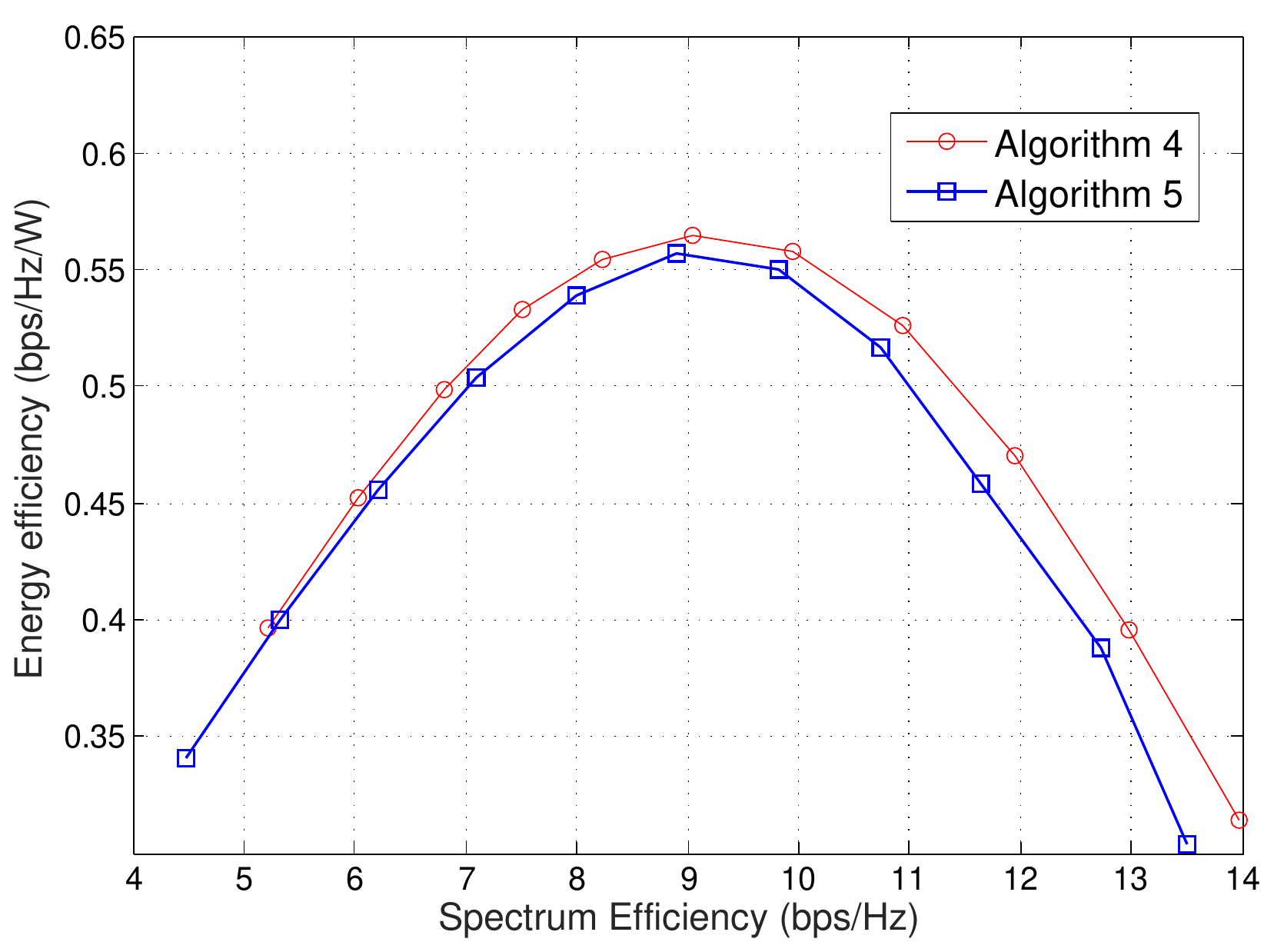}
		\caption{Energy efficiency versus spectral efficiency.}
		\label{fig9}
	\end{center}
\end{figure}

\end{document}